\documentclass{IEEEtran}
\usepackage{setspace}
\usepackage{graphics} % for pdf, bitmapped graphics files
\usepackage{epsfig} % for postscript graphics files
\usepackage{mathtools}
\usepackage{times} % assumes new font selection scheme installed
\usepackage{amsmath} % assumes amsmath package installed
\usepackage{amssymb}  % assumes amsmath package installed
\usepackage{epstopdf}
\usepackage{cite}
\usepackage[final]{changes}
\usepackage{xcolor}
\usepackage{caption}
\usepackage{subfigure}
\usepackage{bbm}
\usepackage{tikz} % Import the tikz package
\usepackage{balance}
\usetikzlibrary{automata, positioning, arrows}

\DeclareMathOperator\arctanh{arctanh}
\DeclareMathOperator\diag{diag}
\usepackage{cite}
\usepackage{amsmath,amssymb,amsfonts}
\usepackage{algorithmic}
\usepackage{graphicx}
\usepackage{textcomp}
\def\BibTeX{{\rm B\kern-.05em{\sc i\kern-.025em b}\kern-.08em
		T\kern-.1667em\lower.7ex\hbox{E}\kern-.125emX}}

\newtheorem{theorem}{Theorem}[section]
\newtheorem{lemma}[theorem]{Lemma}
\newtheorem{proposition}[theorem]{Proposition}
\newtheorem{problem}[theorem]{Problem}

\newtheorem{definition}[theorem]{Definition}
\newtheorem{assumption}[theorem]{Assumption}
\newtheorem{remark}[theorem]{Remark}

\newtheorem{example}[theorem]{Example}

\IEEEoverridecommandlockouts                              % This command is only needed if 
\begin{document}

\title{\LARGE \bf
Co-design of Optimal Transmission Power and Controller for Networked Control Systems Under State-dependent Markovian Channels
}

\author{Bin Hu and Tua A.~Tamba
	\thanks{Bin Hu is with Department of Engineering Technology, Old Dominion University, Norfolk, VA, 23529, US. {\tt\small bhu@odu.edu}}
	\thanks{Tua A.~Tamba is with Department of Electrical Engineering, Parahyangan Catholic University, Bandung, 40141, Indonesia. {\tt\small ttamba@unpar.ac.id}}% <-this % stops a space}
}

\maketitle
%\thispagestyle{empty}
%\pagestyle{empty}

%%%%%%%%%%%%%%%%%%%%%%%%%%%%%%%%%%%%%%%%%%%%%%%%%%%%%%%%%%%%%%%%%%%%%%%%%%%%%%%%
\begin{abstract}
This paper considers a co-design problem for industrial networked control systems to ensure both the stability and efficiency properties of such systems. 
The assurance of such properties is particularly challenging due to the fact that wireless communications in industrial environments are not only subject to shadow fading but also stochastically correlated with their surrounding environments. 
To address such challenges, this paper first introduces a novel state-dependent Markov channel (SD-MC) model that explicitly captures the \emph{state-dependent features} of industrial wireless communication systems by defining the proposed model's transition probabilities as a function of both its environment's states and transmission power. 
Under the proposed channel model, sufficient conditions on \emph{Maximum Allowable Transmission Interval}~(MATI) are presented to ensure both \emph{asymptotic stability in expectation} and \emph{almost sure asymptotic stability} properties of a continuous nonlinear control system with \emph{state-dependent fading channels}. 
Based on such conditions, the co-design problem is then formulated as a constrained polynomial optimization problem (CPOP), which can be efficiently solved using semidefinite programming methods for the case of a two-state state-dependent Markovian channel. 
The solutions to such a CPOP represent optimal control and power strategies that optimize the average expected joint costs in an infinite time horizon while still respect the stability constraints. 
For a general SD-MC model, this paper further shows that sub-optimal solutions can be obtained from  linear programming formulations of the considered CPOP.
Simulation results are given to illustrate the efficacy of the proposed co-design scheme. 
	
\end{abstract}


%%%%%%%%%%%%%%%%%%%%%%%%%%%%%%%%%%%%%%%%%%%%%%%%%%%%%%%%%%%%%%%%%%%%%%%%%%%%%%
\section{Introduction}
\label{sec:intro}

%%% background and motivation, challenges, especially on the state-dependent shadow fading, 
\subsection{Background and Motivation}
Over the past couple of decades, wireless communication technologies have evolved rapidly and became ubiquitous in modern society. 
For instance, wireless communication protocols such as WirelessHart and WiMAX \cite{wang2016implementing,ahlen2019toward} are now considered among the important components which help improve modern industrial automation and have been successfully implemented in various industrial applications for the purpose of building efficient, safe, and reconfigurable industrial automation systems.
%As important components to improve modern industrial automation, wireless communication protocols such as WirelessHart and WiMAX \cite{wang2016implementing,ahlen2019toward}, have been successfully implemented in various industrial applications with the goals of building efficient, safe, and reconfigurable industrial automation systems. 
Building a safe and efficient industrial networked control system (NCS), however, is known to be challenging due to the fact that wireless communication channels in highly dynamic industrial environments are inherently unreliable and often subject to \emph{shadow fading} phenomenon. 
%Shadow fading seriously compromises system stability and performance by causing significant degradation on the quality of communication links. 
The shadow fading phenomenon, in particular, may seriously compromise system stability and performance as it causes significant degradation on the quality of the communication links. 
Moreover, the \emph{shadow fading} effect has also been well known to be correlated to moving objects and machineries in industrial environments \cite{agrawal2014long, qin2018link,ahlen2019toward,quevedo2012state}. 
Such a correlation between industrial NCSs and their external environments poses a great challenge in term of the needs for NCS channel modeling approach as well as assuring the stability and long-term efficiency of industrial automation systems.
This paper addresses such challenges by first proposing a novel \emph{State-Dependent Markov Channel}~(SD-MC) model that specifically incorporates the impact of external environments (e.g. moving objects) on the channel conditions, and then developing a co-design formalism that ensures both the NCS stability and efficiency. 

%due to the motion of large size moving machineries. The body of machineries blocks the radio signal from the transmitter, thereby causing significant degradation in communication performance. To address such challenges, one needs to explicitly incorporate the impact of external environments on the wireless communication conditions in system design. This paper is particularly motivated by such needs, and proposes a novel \emph{state-dependent Markov channel model} to capture the correlation between channel conditions and external environment states. Based on the proposed \emph{state-dependent Markov channel model}, this paper further introduces a co-design framework to jointly consider both system stability and long-term efficiency in the presence of  \emph{shadow fading} communication channels.
\subsection{Related Work}
This subsection will focus on reviewing existing works that are related to (i) wireless channel modeling in industrial environments and (ii) co-design methods for communication and control systems. 
Since it is beyond the scope of this paper to exhaustively overview all related topics, we urge interested readers to also refer to recent survey papers in e.g. \cite{kharb2019survey,hermeto2017scheduling,queiroz2017survey}.

Recent studies have shown that radio communication system in industrial environments often exhibits \emph{shadow fading} effect that is statistically dependent on the environment's various states and dynamics (large metal objects, moving machineries and vehicles) \cite{agrawal2014long,quevedo2012state,quevedo2013power}. 
In particular, the \emph{state-dependent features} of such a fading prevent the use of conventional channel models such as Markov chain \cite{lun2020impact} or identically distributed independent process~(i.i.d.) to capture the dynamics of communication channel in complex industrial environments \cite{ahlen2019toward, qin2018link}. 
%Such \emph{state-dependent features}  prevent conventional channel models, such as Markov chain \cite{lun2020impact} or identically distributed independent process~(i.i.d.), from being applicable to complex industrial environments \cite{ahlen2019toward, qin2018link}. 
In regard to these, research efforts have recently been devoted to developing effective channel models that correlate the temporal variations of the channel conditions with the states of external environment in different industrial settings; see e.g. \cite{ahlen2019toward, qin2018link, quevedo2012state,quevedo2013power,hu2019co}. 
In \cite{quevedo2012state,quevedo2013power}, a network state process modeled by a Markov chain was introduced to characterize the \emph{shadow fading} effects under a finite set of configurations in industrial environments. 
The state-dependent feature of wireless channels was then modeled by defining the probability of packet losses as a function of the network state process. 
Other channel models as reported in \cite{agrawal2014long,qin2018link,olofsson2016modeling} were focused on characterizing the fading statistics that were correlated with the movements  of objects and machines in industrial environments. 
In particular, the work in \cite{qin2018link} proposed a qualitative three-layer impulse responses framework to model the temporal fading effects and showed its effectiveness in capturing  nearby moving objects. 
Alternative methods to channel modeling in industrial environments are based on a combination of multiple probability distributions which captures channel dynamics with different fading parameters~(or states). 
The use of such a mixture of probability distributions are motivated by sudden changes of fading statistics that are observed during extensive channel measurements in various industrial environments \cite{olofsson2016modeling,eriksson2016long,qin2018link,ahlen2019toward,agrawal2014long,vinogradov2015measurement}.  

The SD-MC model proposed in this paper is different from the aforementioned existing models in two aspects. 
First, the works in \cite{quevedo2012state,quevedo2013power} model the external environment~(i.e. a moving vehicle) dynamics as a (semi)-Markov chain and assume that the moving vehicle cannot be controlled. 
This paper removes such uncontrollability assumptions and models the external environment as a Markov Decision Process~(MDP). 
Secondly, the models adopted in \cite{quevedo2012state,quevedo2013power, hu2019co} are confined to packet-drop channels that ignore quantization effects.
In this paper, we consider a generalized state-dependent Markov model that takes into account the presence of time varying data rates which are more realistic in real world industrial settings.
%%% Emphasize the state-dependent feature present in complex industrial environments. One typical example is the the co-existance of the networked system and large moving machinery. Why it is important to incorporate such state-dependent features into the channel model 

In the presence of wireless fading channels, power control as an effective means to mitigate channel fading effects has been well studied in wireless communication community  \cite{goldsmith1997variable}. 
From NCS design perspective, it is important to further ensure both the system stability and efficiency of the whole industrial system. 
These thus suggest that a joint design of power and control strategies must be considered.  
In regard of this, numerous co-design results were developed to design optimal controller \cite{ma2020efficient,zhang2018denial,rabi2016separated,di2019codesign,peters2016controller,zhang2006communication,peng2013event,zhao2018toward} or state estimator \cite{gatsis2014optimal,quevedo2013power,quevedo2012state,ren2017infinite,leong2015kalman,chakravorty2019remote,dolz2017co,qi2016optimal,zhang2016resilient} for NCSs by incorporating the impacts that the fading channels have on the design processes. 
Regarding the co-design of optimal state estimator and transmission power, the works in \cite{gatsis2014optimal,quevedo2013power,quevedo2012state,ren2017infinite,leong2015kalman,chakravorty2019remote,dolz2017co,qi2016optimal,zhang2016resilient} have shown that the optimal estimation strategies can be obtained based on Kalman filters whose structural design is independent of the used wireless communication channels, whereas the optimal power policies take the form of functions of the channel states and the innovation error of the Kalman filter. 
For the co-design of optimal controller, the main ideas in prior works are basically that of applying the so-called separation principle where the optimal design of communication and control strategies can be separated by assuming both systems are independent from each other. 
Such an  independence assumption, however, has limited applicability in complex industrial environments where the communication and control parts of the NCS are tightly coupled with the presence of \emph{state-dependent fading channels}. 
Compared with the currently existing results, the unique feature of the work presented in this paper is the incorporation of mutual interaction between the communication and control systems into the co-design process via the proposed SD-MC model. 
By exploring such state-dependent features, this paper shows that the resulting optimal co-design strategies are more robust and efficient against various levels of \emph{shadow fading} 
than those of the existing conventional methods.  

%This paper differs the aforementioned prior works by developing a novel co-design framework that explores the state-dependent properties of wireless communication channels to ensure both stability and optimal performance of industrial NCSs. 
This paper extends our preliminary results in \cite{hu2020acc} in three main parts. First, this paper includes a new stability result for a weaker notion of \emph{asymptotic stability in expectation} by using a less conservative Assumption \ref{assumption}. Secondly, a linear program-based approximation method is proposed to provide tractable and computationally efficient solutions for the co-design problem. Thirdly, more simulation results are presented in this paper to further demonstrate advantages of the proposed co-design method.

\subsection{Contributions}
The main contributions of this paper are summarized below.

\begin{itemize}
\item A proposal of a novel \emph{state-dependent Markov channel model} that explicitly captures the dependency between channel states, controlled external environments, and transmission power. 
In particular, the incorporation of   \emph{state-dependent features} in the proposed channel model generalizes the traditional Markov chain and i.i.d. models. 
\item Based on the proposed channel model, this paper then derives sufficient conditions on the MATI that will assure \emph{asymptotic stability in expectation} and \emph{almost sure asymptotic stability} properties of a nonlinear NCS. 
In particular, the paper shows that the derived sufficient conditions on MATI generalizes existing results in \cite{nesic2004input,nesic2009explicit,heijmans2017computing} through the incorporation of the \emph{state-dependent} properties in the conditions. 
\item Using the derived MATI constraints, this paper then proposes a constrained polynomial optimization problem (CPOP) formulation to solve the co-design problems.
Under the proposed co-design problem, the system stability is assured by imposing the derived sufficient conditions as hard constraints in the CPOP formulation. 
The solutions to the CPOP thus represent optimal control and transmission power policies that minimize an average joint costs for both communication and control systems.
\item The next contribution of the paper is the development of efficient algorithms to solve the co-design problems by showing that the formulated CPOP can be efficiently solved using SDP methods if a two-state Markovian channel model is considered. 
For a general Markovian channel model, the formulated CPOP can also be approximated as linear programing (LP) problems whose solutions lead to sub-optimal co-design strategies. 
\item Finally, this paper also presents extensive simulation results which demonstrate the benefits/advantages of adopting the proposed co-design framework against various levels of \emph{shadow fading} that are commonly encountered in complex industrial environments.  
\end{itemize}

The paper is structured  as follows. 
A detailed description of the system framework is given in Section \ref{sec:system-framework} and then followed by  problem formulation in Section \ref{sec:problem-formulation}. 
The main results in terms of sufficient conditions on MATI and optimal co-design strategies are provided in Section \ref{sec:main-results}. 
Simulation results are presented in Section \ref{sec:simulation}. Section \ref{sec:conclusion} concludes the paper. 

\emph{Notations:} Throughout the paper, let $\mathbb{R}$, $\mathbb{Z}$ denote the sets of real and integer numbers, respectively, and $\mathbb{R}_{\geq 0}, \mathbb{Z}_{\geq 0}$ denote their non-negative counterparts. 
Let $\mathbb{R}^{n}$ and $\mathbb{R}^{n \times m}$ denote the $n$-dimensional real vector space and matrix of dimension $n \times m$, respectively. 
For a vector $ x \in \mathbb{R}^{n}$, let $|x|=\max_{i}|x_{i}|$ denotes its infinity norm, where $x_{i}$ is the $i$-th element of the vector and $1 \leq i \leq n$.
For a matrix $A \in \mathbb{R}^{n \times m}$, let $|A|:=\|A\|_{\infty}=\max_{1\leq i \leq m}\sum_{j=1}^{n}|a_{ij}|$ denotes its infinity norm. 

%    Increasing transmission power may compensate for system performance loss caused by channel fading, but may also over-consume limited resources of sensor networks in the long run if the transmission power strategy is not appropriately designed. This issue becomes more complicated and challenging in industrial process systems where the fading is correlated with the external environments. 
%%% The wireless communication channels are subject to shadow fading 

%%% Discuss the relevant research work and talk about the difference between the existing work and the one in this paper. The necessity to bring the co-design scheme to ensure system efficiency in the long run 

%%% Briefly layout the contributions of this paper. (1) state-dependent Markov channel model, (2) co-design framework that addresses both the stability and efficiency problems (3) shows that co-design problem can be solved by a constrained polynomial optimization or a SDP program if a two-state state-dependent Markov channel is considered. 

%%% include some notations, | |, || ||, infinity norm 
\section{System Framework}
\label{sec:system-framework}
This paper considers a system framework as shown in Fig. \ref{fig: sys} which consists of a \emph{nonlinear plant}, a SD-MC $\mathcal{M}_{c}$, a \emph{remote controller}, and \emph{external environment}  modeled by MDP $\mathcal{M}$.
\begin{figure}[!t]
	\centerline{\includegraphics[width=\columnwidth]{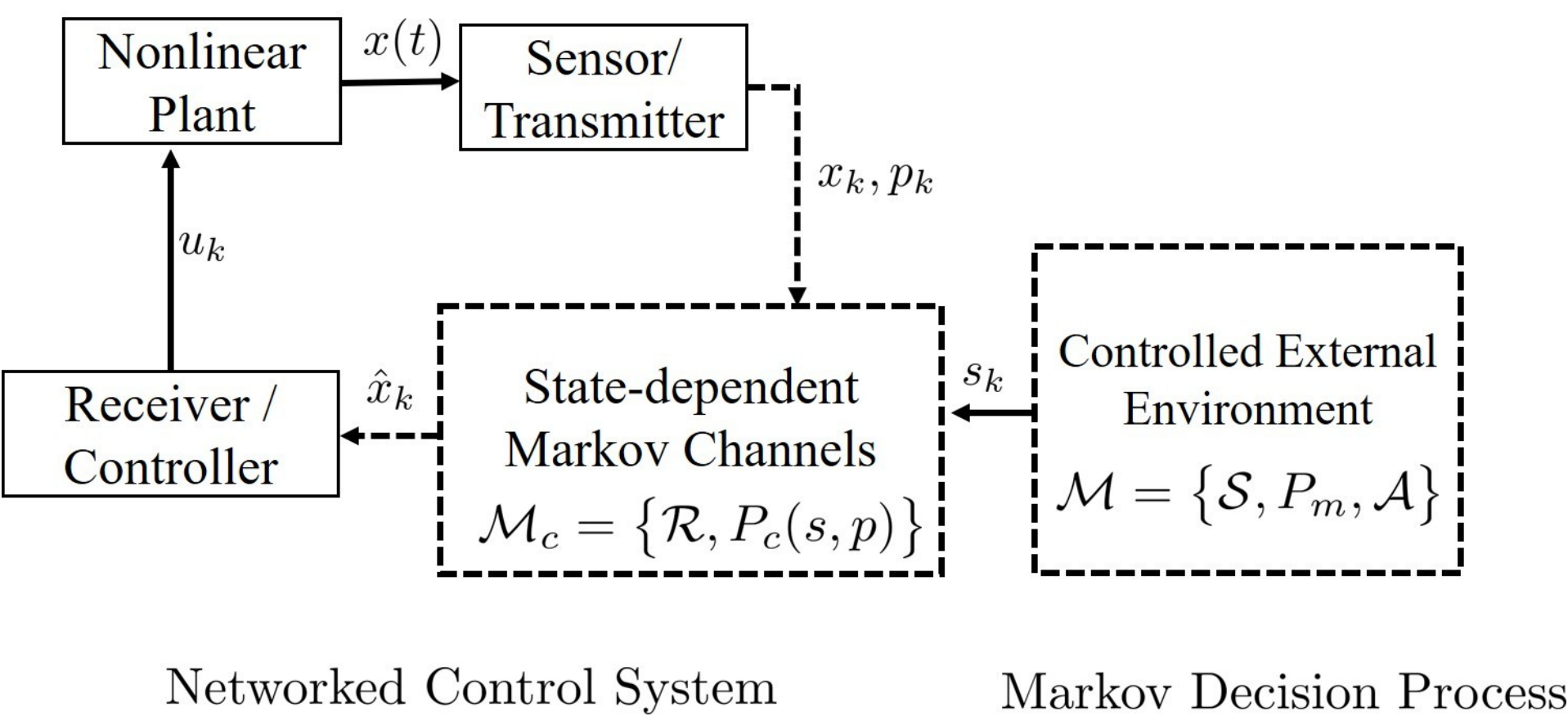}}
	\caption{Nonlinear NCS with SD-MC framework.}
	\label{fig: sys}
\end{figure}

%\subsection{System Model}
\subsection{Nonlinear Plant} The nonlinear plant dynamics are modeled by an ODE:
\begin{align}
	\label{sys:nonlinear}
	\dot{x}=f(x, u)
\end{align}
where $x \in \mathbb{R}^{n_{x}}$ is the plant/system state and  $u \in \mathbb{R}^{n_u}$ is the system input  that is generated by a remote controller. 
The system vector fields are governed by a nonlinear function $f: \mathbb{R}^{n_{x}} \times \mathbb{R}^{n_u}  \rightarrow \mathbb{R}^{n_{x}}$ that is locally Lipschitz with respect to $x$. 
In this paper, we assume that the system state $x$ can not be accessed  directly by the controller, and thus must be transmitted through a wireless communication channel. 

\subsection{Sensor and Transmitter} The system state $x$ is first pre-processed by a \emph{sensor/transmitter} module before the transmission. This paper considers both the sampling and quantization effects on $x$. 
Specifically, let $\{t_{k}\}_{k=0}^{\infty}$ be a sequence of sampling time instants with $t_{k} < t_{k+1}$, and $x_{k}:=x(t_{k})$ be the sampled state value at time instant $t_{k}$. 
The sampled state $x(t_k)$ is then encoded by one of a finite number of symbols that are constructed based on a dynamic quantization scheme \cite{nesic2009unified}. 
Such a dynamic quantization scheme maps the sampled state $x \in \mathbb{R}^{n_x}$ into the index of a finite number of symbols. 
Specifically, let $R \in \mathbb{N}$ denote the number of bits used to construct the symbol such that the sequence of the symbols can then labeled as $\mathcal{S}=\{1, 2, \ldots, 2^{R}\}$ with a total number of $2^{R}$. 
The way of using such symbols to encode the state information $x(t_k)$ is by constructing a dynamic quantizer that is able to track the evolution of $x(t_k)$ at each transmission instant. 
The quantizer is defined as a tuple $\mathcal{Q}=(\mathcal{S}, q(\cdot), \xi)$ where $q(\cdot):\mathbb{R}^{n_x} \rightarrow \mathcal{S}$ is a quantization function which maps the system state into the symbol, while $\xi \in \mathbb{R}_{\geq 0}$ is an auxiliary variable defining the size of the quantization regions. 

A typical approach to implement the quantizer $\mathcal{Q}$ is based on the construction of a hybercubic box which evolves dynamically to contain and track the state $x$. 
To demonstrate the mechanism of such a box-based dynamic quantizer, suppose a hypercubic box is constructed at time instant $t_k$.
Let $\hat{x}(t_k)$ and $2\xi(t_k)$, respectively, denote the center and the size of such a box. 
Then, the box is divided equally into $2^{R}$ smaller sub-boxes with each sub-box is then labeled as one of the symbols in $\mathcal{S}$. 
Among all the symbols, let $q(x) \in \mathcal{S}$ denotes the symbol (sub-box) that contains the state information $x$. 
The center of that sub-box, $\hat{x}(t_k^{+})$, is then used as an updated estimate of the state information $x$ at time instant $t_k$. 
Thus, te hyper-cubic box may then updated with a new center of $\hat{x}(t_{k}^{+})$ and a new size of $\xi(t_{k}^{+})=\xi(t_{k})/2^{R}$. 
The symbol representing this updated hyber-cubic is transmitted through the wireless communication channel. 
In these regards, the dynamics of such a quantizer can be characterized by the following equations.
%The following equations are used to characterize the  
\begin{subequations}
	\label{sys:quantizer}
	\begin{align}
		\hat{x}(t_{k}^{+})&=h(k, q(x(t_k)), \hat{x}(t_k), \xi(t_k), R_k) \\
		\xi(t_{k}^{+})&=\frac{\xi(t_{k})}{2^{R_k}}
	\end{align}
\end{subequations}
where $R_{k}$ is the number of bits available at time instant $t_k$ in the form of a time varying variable which depends on the wireless channel conditions in real time. 
As discussed in prior work \cite{nesic2009unified,liberzon2005stabilization}, within each time interval $[t_k, t_{k+1}), \forall k \in \mathbb{Z}_{\geq 0}$, the size of the hyper-cubic box needs to be propagated to ensure that the constructed box captures the actual state $x$. 
Thus, to be more specific, we define the following differential equation to characterize the evolution of the box size over time.
\begin{align}
	\dot{\xi}(t)=g_{\xi}(\xi), \forall t \in  [t_k, t_{k+1})
	\label{eq:box-size}
\end{align}

\subsection{Remote Controller} Under the assumed dynamic quantizer $\mathcal{Q}$, this paper considers a model-based remote controller that maintains a "copy" of the plant dynamics in \eqref{sys:nonlinear} as defined as follows.
\begin{align}
	\label{sys:controller}
	\dot{\hat{x}}&=f(\hat{x}, u), \nonumber \\
	u&=\kappa(\hat{x}),  \forall t \in [t_k, t_{k+1}) 
\end{align}
with initial state $\hat{x}(t_{k})=\hat{x}(t_{k}^{+})$ in the time interval $[t_k, t_{k+1})$. The control function $\kappa(\cdot):\mathbb{R}^{n_x} \rightarrow \mathbb{R}^{n_u}$ is a "nominal" controller that is selected to stabilize the dynamic system in \eqref{sys:controller} without considering the effect of the network.

Based on the system dynamics in \eqref{sys:nonlinear}, the dynamic quantizer in \eqref{sys:quantizer} and \eqref{eq:box-size}, and the remote controller in \eqref{sys:controller}, the closed loop system can be characterized as a stochastic hybrid system defined as below:
\begin{subequations}
	\label{sys:continuous}
	\begin{align}
		\dot{x}&=\tilde{f}(x, e) \\
		\dot{e}&=g_{e}(x, e) \\
		\dot{\xi}&=g_{\xi}(\xi), \forall t \in (t_k, t_{k+1}) 
	\end{align}
\end{subequations}
and 
\begin{subequations}
	\label{sys:jump}
	\begin{align}
		e(t_{k}^{+})&=J_{e}(k, x(t_k), e(t_k), \xi(t_k), R_k) \\
		\xi(t_{k}^{+})&=J_{\xi}(\xi(t_k), R_k), \forall k \in \mathbb{Z}_{\geq 0}
	\end{align}
\end{subequations}
where $e:=x-\hat{x}$ denotes the estimation error and $\tilde{f}(x, e)=f(x, \kappa(x-e))$, $g_{e}(x, e)=f(x, \kappa(x-e))-f(x-e, \kappa(x-e))$, $J_{e}(k, x(t_k), e(t_k), \xi(t_k), R_k)=x(t_k)-h(k, q(x(t_k)), x(t_k)-e(t_k), \xi(t_k), R_k)$, and $J_{\xi}=\xi(t_k)2^{-R_k}$. 
Model \eqref{sys:continuous} basically characterizes the continuous dynamics of the closed loop system while model \eqref{sys:jump} describes the stochastic jump behavior of the system under the resulting time varying data rate of $R_k$.

\subsection{Controlled External Environments $\&$ SD-MC Models} 

The external environments in industrial settings which consist of moving vehicles or machines are modeled by a MDP $\mathcal{M}_{env}=\{S, s_{0}, A, Q\}$ where $S=\{s_{i}\}_{i=1}^{M_s}$ is a finite set of environment states, $s_0$ is an initial state, $A=\{a_{i}\}_{i=1}^{M_a}$ is a finite set of actions, and $Q=\{q(s|s', a)\}_{s, s' \in S, a \in A}$ is a transition matrix. 
Considering as an example a forklift vehicle which operates in an industrial factory, then the state set $S$ in the MDP represents a group of partitions of its operating regions on the factory floor. 
By taking a particular action $a \in A$, the forklift thus moves from one region $s'$ to another $s$ following a transition probabilities $q(s|s', a)$. 
Under the considered environment model, the quality of the wireless communication link is affected by the state/region where the forklift vehicle is located. 
The quality of the communication link is in particular measured by a time varying data rate that is selected from a finite set $\mathcal{R}=\{r_{1}, r_{2}, \ldots , r_{M_{R}}\}$. 

Let the sequence $\mathcal{I}=\{t_{k}\}_{k=0}^{\infty}$ denotes the transmission time instants while the sequence of random variables $R_{k}$ denote the data rates at time instant $t_k$.
 Then,  the sequence $\{R_{k}\}_{k=0}^{\infty}$ over these transmission time instants form a random process which characterizes the stochastic variations on the channel conditions. 
 At each time instant $t_k$, the communication system can adjust its transmission power level for sending data through a wireless communication channel. 
 Let $\Omega_{p}=\{1, 2, \ldots, M_p \}$ denotes a finite set of transmission power levels wherein each $i \in \Omega_{p}$ represents the power level $i$. 
 The transmission power set is sorted in an ascending order such that larger numbers represent higher power levels. 
 Let $p_{k}:=p(t_{k}) \in \Omega_{p}$ denotes the power level that is selected at time instant $t_k$. 
 With the assumed MDP model for the external environment, we define the SD-MC model below.
 
\begin{definition}
	\label{def:ss-MC}
	Given a power set $\Omega_{p}$, an external environment modeled by a MDP $\mathcal{M}_{env}$, and a finite set of data rates $\mathcal{R}=\{r_{1}, r_{2}, \ldots , r_{M_{R}}\}$ that are arranged ascendingly, i.e., $r_{i} < r_{j}, \forall i < j$.
	Then a  wireless communication channel is said to be a SD-MC if $\forall s \in S, p\in \Omega_{p}$, and $\forall r_{i}, r_{j} \in \mathcal{R}$:
	\begin{align}
		\label{sys:ss-MC}
		\mathbb{P}\{R_{k+1}=r_i | R_{k}=r_{j}, s_{k}=s, p_{k}=p\}=P_{ij}(s, p)
	\end{align}
	where $P_{ij}(s, p)$ is a transition probability from data rate $r_{j}$ to $r_{i}$ given the transmission power $p$ and environment state $s$. 
\end{definition} 

The SD-MC model in \eqref{sys:ss-MC} can be viewed as a generalization of more conventional Markov channel model which ignores the impact of environment state and transmission power \cite{minero2012stabilization,zhang1999finite}. 
Example \ref{example-ss-MC} below illustrates how the SD-MC model in \eqref{sys:ss-MC} can be derived using classical estimation methods in \cite{zhang1999finite,goldsmith1997variable} which were often used to generate the conventional Markov channel model.
 From a technical standpoint, Example \ref{example-ss-MC} thus demonstrates that the SD-MC  model proposed in this paper generalizes the conventional Markov channel model.
\begin{example}
	\label{example-ss-MC}
	The characterization of the SD-MC model in \eqref{sys:ss-MC} is an abstraction of the temporal channel state variations under different transmission power and external environment dynamics. 
	The received signal-to-noise ratio (SNR) is a function of the transmission power and the environment states. 
	To explicitly show the impact of such state-dependent features on the received SNR, suppose that the (small scale)~channel gain follows a Raleigh distribution with $Rayleigh(1)$\footnote{The example selects unit parameter $Rayleigh(1)$ for the sake of notation simplicity and the channel model proposed in \eqref{sys:ss-MC} can be applied to other parameters and distributions.}.
	Let $\varphi: S \rightarrow [0, 1]$ denotes a function which characterizes the impact of different environment states $s \in S$ on the shadow fading effect.
	Then, the received instantaneous SNR $\gamma$ under the transmission power $p \in \Omega_{p}$ and the environment state $s \in S$ is a random variable that is exponentially distributed with a probability density function of the form 
	\begin{align}
		f_{d}(\gamma)=\frac{1}{\overline{\gamma}}\exp\left(-\frac{\gamma}{\overline{\gamma}}\right), \gamma \geq 0
	\end{align}
	where $\overline{\gamma}=\frac{p\varphi(s)}{N_{0}B}$ is the average SNR with noise density of $N_0$ and received signal bandwidth of $B$. 
	The states of $M_{R}$ data rates in the Markov channel are obtained by partitioning the received SNR into $M_{R}$ disjoint intervals $\cup_{i=1}^{M_{R}} [\Gamma_i, \Gamma_{i+1})$ with $\Gamma_{i}$ and $\Gamma_{i+1}$ being the lower and higher thresholds, respectively, of the $i$-th interval. 
	For each interval, there is one selected data rate $r \in \mathcal{R}$ to ensure a sufficiently small bit error rate \cite{zhang1999finite}. 
	Thus, we say that the data rate is $r_i$ if the received SNR values are located in the corresponding interval. 
	Based on the results in \cite{zhang1999finite}, the corresponding transition probability may then be approximated as follows.
	\begin{equation}
	\label{exp-tp}
	\begin{aligned}
		P_{i, i+1}&=\frac{N_{i+1}T_p}{\pi_{i}}, P_{i,i-1}=\frac{N_{i}T_p}{\pi_{i}},\\
		 P_{ii}&=1-P_{i, i+1}-P_{i,i-1}
	\end{aligned}
	\end{equation}
	where\footnote{The parameter $f_{D}$ is the maximum Doppler frequency and $T_p$ is the one packet time period.} $N_{i}=\sqrt{\frac{2\pi \Gamma_{i}}{\overline{\gamma}}}f_{D}\exp\left(-\frac{\Gamma_{i}}{\overline{\gamma}} \right)$ and $\pi_{i}={\rm Pr}\{\Gamma_{i} \leq \gamma \leq \Gamma_{i+1}\}=\exp\left(-\frac{\Gamma_i}{\overline{\gamma}}\right)-\exp\left(-\frac{\Gamma_{i+1}}{\overline{\gamma}}\right)$. Since $\overline{\gamma}=\frac{p\varphi(s)}{N_{0}B}$, the transition probabilities defined in \eqref{exp-tp} change as a function of the transmission power $p$ and the environment states $s$. 
\end{example}

%=================================================================================================================

\section{Problem Formulation}
\label{sec:problem-formulation}
For the closed-loop NCS modeled in \eqref{sys:continuous}, \eqref{sys:jump} and \eqref{sys:ss-MC}, one of the problems considered in this paper is that of ensuring the system's  stochastic stability property. 
 Definition \ref{def: ss} formally states the \emph{stochastic stability} notions used in this paper.
\begin{definition}[Stochastic Stability \cite{khasminskii2011stochastic}]
	\label{def: ss}
	Given the SD-MC model in \eqref{sys:ss-MC} and the closed-loop NCS model in \eqref{sys:continuous}-\eqref{sys:jump}.
	\begin{itemize}
		\item[\textbf{E}] The system is said to be \emph{asymptotically stable in expectation} (ASE) if there exist a class $\mathcal{KL}$ function $\beta(\cdot, \cdot)$ and a bounded set $\Omega_{r}=\{ x \in \mathbb{R}^{n} : |x| <  r \}$ such that for all $x_{0} \in \Omega_{r}$, it holds that:
		\begin{align}
			\label{ineq:ASE}
			\mathbb{E}(|x(t)|) \leq \beta(|x_0|, t-t_0),   \quad \forall t \in \mathbb{R}_{\geq 0}
		\end{align}
		and that $\lim_{t \rightarrow \infty}\mathbb{E}(|x(t)|)=0$.
		\item[\textbf{P}] The system is said to be \emph{almost surely asymptotically stable} (ASAS) if $\forall \epsilon, t' > 0$, the following holds:
		\begin{align}
			\label{ineq:ASAS}
			\mathbb{P}\{\lim_{t' \rightarrow \infty}\sup_{t \geq t'}|x(t)| \geq \epsilon \} = 0
		\end{align}
	\end{itemize}
\end{definition}
\begin{remark}
	The ASAS is a stronger stability notion than that of the ASE in the sense that (i) the former implies the latter while (ii) the latter normally does not lead to the former. 
	As a matter of fact, one may show that a system is ASAS if it is \emph{exponentially stable in expectation}, i.e. there exists an exponential function $\beta(|x_0|, t-t_0):=c_{1}\exp\big(-c_{2}(t-t_{0})\big)|x_0|$ with $c_{i} > 0,  i=1, 2$ such that inequality \eqref{ineq:ASE} holds \cite{hu2019co, hu2019optimal}. 
\end{remark}
\begin{problem}[Stability Problem]
	\label{problem-stability}
	The first problem considered in this paper is the stability problem.
	Specifically, the problem is to characterize the MATI under which the considered NCS with the proposed SD-MC model satisfies the stochastic stability  notions given in Definition \ref{def: ss}.
\end{problem}
\begin{problem}[Co-design Problem]
	\label{problem-efficiency}
Under the stability conditions obtained by solving Problem \ref{problem-stability}, the second problem is to find optimal control and transmission power policies under which some pre-defined system costs are optimized. 
Specifically, let $\mu_{m}$ and $\mu_{p}$ denote the control policy and transmission power policy respectively.
Then, for a given joint cost function $\{c(s, p, r)\}_{s \in S, p \in \Omega_{p}, r \in \mathcal{R}}$ the co-design problem is defined as that of finding an optimal policy $\mu^{\star}:=(\mu^{\star}_{m}, \mu^{\star}_{p})$ such that the average expected costs in \eqref{opt: game-1} below are minimized under the obtained stability conditions.
	\begin{equation}
		\label{opt: game-1}
		\begin{aligned}
			& \underset{\mu_{m}, \mu_{p}}{\text{min}}
			%			& & J(\{s_{0}, R_{0}\}, \mu_{m}, \mu_{p})=\lim_{T \rightarrow \infty}\frac{1}{T} \mathbb{E}_{s_{0}, R_{0}}^{\mu_{m}, \mu_{p}} \sum_{i=0}^{T} c(s_{k}, p_{k}, R_{k})\\
			& & \lim_{\ell \rightarrow \infty}\frac{1}{\ell} \mathbb{E}_{s_{0}, R_{0}}^{\mu_{m}, \mu_{p}} \sum_{k=0}^{\ell} c(s_{k}, p_{k}, R_{k})\\
			& \text{s.t.}
			& & \text{Stability conditions ensuring } \eqref{ineq:ASAS}
		\end{aligned}
	\end{equation}
	where $s_{0}$ and $R_{0}$ denote the initial states of the MDP system and the Markov channel, respectively. 
\end{problem}

%This paper first derives sufficient conditions on MATI to ensure stochastic stability notions defined in Definition \ref{def: ss}. These stability conditions serve as constraints in the optimization problem formulated in \eqref{opt: game-1}.
\section{Main Results}
\label{sec:main-results}
This section presents the main results of the paper. 
Firstly, sufficient conditions on MATI which ensure the ASE~(\textbf{E} in Definition \ref{def: ss})) and ASAS~(\textbf{P} in Definition \ref{def: ss})) properties to hold are derived. 
Subsequently, this section then shows that the co-design problem can be solved by constrained optimizations with the obtained stability conditions as constraints.

Assumption \ref{assumption} below is used to derive the main results. 
\begin{assumption}[\cite{nesic2009explicit,carnevale2007lyapunov}]
	\label{assumption}
Consider the stochastic hybrid system in \eqref{sys:continuous} and \eqref{sys:jump}. 
Let $\overline{e}:=[e; \xi]$ be an augmented vector of the error states $e$ and the size of the dynamic quantizer $\xi$. 
Suppose there exist:
\begin{itemize}
\item a function $W : \mathbb{N}_{\geq 0} \times \mathbb{R}^{n_{e}+1} \rightarrow \mathbb{R}_{\geq 0}$ that is locally Lipschitz with respect to $\overline{e}$
\item a function $V: \mathbb{R}^{n_x} \rightarrow \mathbb{R}_{\geq 0}$ that is  locally Lipschitz, positive definite, and radially unbounded
\item a continuous function $H: \mathbb{R}^{n_x} \rightarrow \mathbb{R}_{\geq 0}$
\end{itemize}
Assume further that there exist a finite set of constants $\{\lambda_{i}\}_{i=1}^{M_{R}}$, some real numbers $L \geq 0, \zeta > 0$, class $ \mathcal{K}_{\infty}$ functions $( \underline{\alpha}_{W}, \overline{\alpha}_{W}, \underline{\alpha}_{V}, \overline{\alpha}_{V}) \in \mathcal{K}_{\infty}$, and a continuous positive definition function $\varrho$ such that:
\begin{enumerate}
	\item $\forall k \in \mathbb{N}$, $\overline{e} \in \mathbb{R}^{n_{x}+1}$ and $r_{i} \in \mathcal{R}=\{r_{1}, \ldots, r_{M_R}\}$:
	\begin{subequations}
		\label{ineq:W}
		\begin{align}
		\underline{\alpha}_{W}(|\overline{e}|) \leq W(k, \overline{e}) &\leq \overline{\alpha}_{W}(|\overline{e}|) \\
		W(k+1, \overline{J}(k, \overline{e}, r_i)) &\leq \lambda_{i} W(k, \overline{e})
		\end{align}
	\end{subequations}
where $\overline{J}(k, \overline{e}, r_i)=[J_{e};J_{\xi}]$ with functions $J_{e}$ and $J_{\xi}$ are defined as in \eqref{sys:jump}.
\item $\forall k \in \mathbb{N}, x \in \mathbb{R}^{n_x}$ and for almost all $\overline{e} \in \mathbb{R}^{n_{x}+1}$,  then 
\begin{align}
\label{ineq:W-continuous}
\left\langle \frac{\partial W(k, \overline{e})}{\partial \overline{e}}, \overline{g}(x, \overline{e})  \right\rangle \leq L W(k, \overline{e})+ H(x)
\end{align}
where $\overline{g}(x, \overline{e})=[g_{e};g_{\xi}]$ with functions $g_{e}$ and $g_{\xi}$ are defined as in \eqref{sys:continuous}.
	\item $\forall x \in \mathbb{R}^{n_x}$ 
	\begin{align}
	\label{ineq: alpha_V}
	\underline{\alpha}_{V}(x) \leq V(x) \leq \overline{\alpha}_{V}(x)
	\end{align}
	and $\forall k \in \mathbb{N}, \overline{e} \in \mathbb{R}^{n_{x}+1}$, and for almost all $x \in \mathbb{R}^{n_x}$
	\begin{align}
	\label{ineq:delta-V}
	\langle \nabla V(x),  \tilde{f}(x, e) \rangle \leq \varrho(|x|)-&\varrho(W(k,\overline{e}))-H^{2}(x) \nonumber  \\
	+& \zeta^{2}W^{2}(k, \overline{e})
	\end{align}
\end{enumerate}
\end{assumption}  
\begin{remark}
Assumption \ref{assumption} is essentially similar to \cite[Assumption 1]{carnevale2007lyapunov} where the inequalities \eqref{ineq:W} of  part 1) are used to characterize the bounds on the function of error states $\overline{e}$ as well as its growths for different data rates at discrete time instants. 
It is assumed that, for any given data rate $r_{i} \in \mathcal{R}$, there exists a corresponding positive real $\lambda_{i}$ that bounds the growth of the error function from the above. The inequality \eqref{ineq:W-continuous} of part 2) assumes a linear growth of the error function in the continuous time domain. The inequalities \eqref{ineq: alpha_V} and \eqref{ineq:delta-V} of part 3) are used to characterize the growth rate of the Lyapunov function with respect to the state $x$ in the continuous time domain. The MATI bounds that ensure stochastic stability will be derived based on the parameters given in this assumption.
\end{remark}
\subsection{Sufficient Conditions for Stochastic Stability}
\label{sec:almost-sure-stability}
Under Assumption \ref{assumption}, this subsection derives sufficient conditions on the MATI that will ensure the stability of the stochastic hybrid system in \eqref{sys:continuous}-\eqref{sys:jump} under the proposed SD-MC model.
Specifically, sufficient conditions which ensure the ASE and ASAS properties to hold are given in Theorem \ref{thm:ase} and Theorem \ref{thm:asas}, respectively. 
\begin{theorem}
	\label{thm:ase}
Consider the stochastic hybrid system in \eqref{sys:continuous}-\eqref{sys:jump}, the SD-MC model in \eqref{sys:ss-MC}, and the controlled external environment  $\mathcal{M}_{env}$. Suppose Assumption \ref{assumption} holds.
Let $T_{MATI}$ denotes the maximum allowable transmission time interval.
For a given joint policy $\mu=(\mu_{m}, \mu_{p})$, the system is ASE if 
\begin{align}
T_{MATI} \leq \begin{cases}
\frac{1}{L\eta}\arctan \left(\frac{\eta(1-\overline{\lambda})}{2\frac{\overline{\lambda}}{1+\overline{\lambda}}(\frac{\zeta}{L}-1)+1+\overline{\lambda}} \right) \quad & \zeta > L \\
\frac{1}{L}\frac{1-\overline{\lambda}}{1+\overline{\lambda}} \quad & \zeta=L \\
\frac{1}{L\eta}\arctanh \left(\frac{\eta(1-\overline{\lambda})}{2\frac{\overline{\lambda}}{1+\overline{\lambda}}(\frac{\zeta}{L}-1)+1+\overline{\lambda}} \right) \quad &\zeta < L
\end{cases}
\label{ineq:MATI}
\end{align}
with $\eta=\sqrt{\left|(\frac{\zeta}{L})^{2} -1 \right|}$ and $\overline{\lambda}$ is a constant which satisfies
\begin{align}
\label{ineq:condition-lambda}
\overline{\lambda} > \sqrt{\|\diag(\lambda_{i}^{2})\overline{P}(\mu)\|}
\end{align}
where $\overline{P}(\mu)=\left[\overline{P}_{ij}(\mu)\right]_{1 \leq i, j \leq M_{R}}$ with
$\overline{P}_{ij}(\mu)
=\sum_{p \in \Omega_{p}, s \in S} \mathbb{P}(r_{i} \vert r_{j}, s, p)\mathbb{P}( s, p \vert r_{j}).$
\end{theorem}
\begin{proof}
	The proof is given in Appendix \ref{appendix}.
\end{proof}
\begin{remark}
The MATI bounds shown in \eqref{ineq:MATI} are functions of  parameters $\xi$ and $L$ defined in Assumption \ref{assumption}, and $\overline{\lambda}$ that depends on the parameters of the SD-MC model in \eqref{sys:ss-MC}.  
The proposed MATI bounds differ from the existing results in \cite{nesic2009explicit, carnevale2007lyapunov, hu2019co, hu2019optimal} in two aspects. 
First, the MATI bounds in \eqref{ineq:MATI} generalizes the results in \cite{nesic2009explicit, carnevale2007lyapunov} as they take into account the impacts that the stochastic communication channel has on the MATI. 
Such an impact is captured by the choice of parameter $\overline{\lambda}$ that must satisfies inequality \eqref{ineq:condition-lambda}. 
Existing results may thus be recovered from our MATI bounds by setting the parameter $\overline{\lambda}$ to be independent of the channel conditions. 
Second, the MATI results in this paper extend our prior works in \cite{hu2019co, hu2019optimal} by considering a less conservative assumption on the system structure and a more general SD-MC model. 
\end{remark}

\begin{theorem}
	\label{thm:asas}
Suppose conditions \eqref{ineq:MATI} and \eqref{ineq:condition-lambda} hold for the stochastic hybrid system in \eqref{sys:continuous}-\eqref{sys:jump}. 
Assume there exist positive constants $\underline{\alpha}_{W}, \overline{\alpha}_{W}, \underline{\alpha}_{V}, \overline{\alpha}_{V}$, and $\varrho > 0$ such that the conditions in  \eqref{ineq:W}, \eqref{ineq: alpha_V} and \eqref{ineq:delta-V}  in Assumption \ref{assumption} hold with 
\begin{subequations}
	\begin{align}
		\underline{\alpha}_{W} |\overline{e}| &\leq W(k, \overline{e}) \leq \overline{\alpha}_{W} |\overline{e}| \\
			\underline{\alpha}_{V} |x|^2 &\leq V(x) \leq \overline{\alpha}_{V} |x|^2 \\
		\langle \nabla V(x),  \tilde{f}(x, e) \rangle &\leq \varrho |x|^{2}-\varrho W(k,\overline{e}) \nonumber \\ 
		 & \quad -H^{2}(x)+\zeta^{2}W^{2}(k, \overline{e}). 
	\end{align}
	\label{ineq:new}
\end{subequations}
Then the stochastic hybrid system in \eqref{sys:continuous}-\eqref{sys:jump} is ASAS.
\end{theorem}
\begin{proof}
	The proof is included in Appendix \ref{appendix}.
\end{proof}

\subsection{Optimal Co-design of Control \& Power Policies: A Constrained Optimization Problem}
\label{sec:noncooperative-game}
This section formulates the co-design problem of control and power policies as a constrained optimization problem with the stability condition in \eqref{ineq:condition-lambda} as its constraints. 
The challenge of solving the formulated optimization problems lies in the difficulty of dealing with the stability constraint. 
In particular, the stability condition derived in \eqref{ineq:condition-lambda} is equivalent to polynomial constraints in the optimization problem. 

%\begin{proposition}
%\label{prop:constraint-poly}
%Let $\mu=\{\mathbb{P}(s, p | r)\}_{s \in S, p \in \Omega_{p}, r \in \mathcal{R}}$ denote a given joint policy, the stability constraint in \eqref{ineq:condition-lambda} can be reformulated as $M_R$ polynomial inequalities formulated as below, $\forall 1\leq i \leq M_R$
%\begin{align}
%\label{ineq:poly}
%\sum_{j=1}^{M_R}\sum_{s, p}P_{ji}(s, p)X(s, p)\prod_{\ell \neq j}X(r_{\ell})-\theta_{i}^2\prod_{j=1}^{M_R}X(r_j) \leq 0
%\end{align}
%where $X(s,p)=\sum_{j=1}^{M_R}X(r, s, p)$, $X(r)=\sum_{s, p}X(r, s, p)$\footnote{$X(A)$ denote the occupation~(probability) measure \cite{altman1999constrained} that assigns a probability to the event $A$}, $\theta_{i}=\overline{\lambda}\big/ \lambda_{i}$, and $P_{ji}(s, p)$ is the transition probability defined in \eqref{sys:ss-MC}. Moreover, the given joint policy can be constructed by 
%\begin{align}
%\mathbb{P}(s, p |r)=\frac{X(s, p, r)}{\sum_{s, p}X(s, p, r)}, \quad \forall s \in S, p \in \Omega_{p}, r \in \mathcal{R}
%\end{align}
%\end{proposition}
%\begin{proof}
%	The proof is provided in Appendix \ref{appendix}.
%\end{proof}
%\begin{remark}
%
%\end{remark}
Theorem \ref{thm:qcqp} below shows that, if stationary policies are considered, the co-design Problem \ref{problem-efficiency} can be reformulated as a polynomial constrained program with a linear objective. 
\begin{theorem}
	\label{thm:qcqp}
Let the sets of MDP states $S$, transmission power $\Omega_{p}$, and data rate $\mathcal{R}$ be given.
For all $i$ with $ 1\leq i \leq M_R$, let $X(s, r, p) \geq 0, \forall s, r, p$ denotes the decision variables of the following constrained polynomial optimization problem (CPOP).
%Consider the following polynomial constrained optimization problem, for given sets of 
	\begin{subequations}
\label{opt: qcqp}
\begin{alignat}{5}
&\!\min_{\{X(s, r, p)\}} & & \sum_{s \in S, p \in \Omega_{p}, r \in \mathcal{R}} c(r,s,p)X(r,s,p) \label{opt: obj}\\
& \text{s.t.}
& & \sum_{s, p} X(s, r_i, p) \nonumber \\
	& & & \;\;-\sum_{p, s, r_j} P_{ij}(s, p)X(r_j, s, p)=0, \; \forall r_i \in \mathcal{R} \label{opt: mdp2} \\
& & &\sum_{s , p, r}X(s, r, p)=1,  \label{opt: occupation} \\
& & &\sum_{j=1}^{M_R}\sum_{s, p}P_{ij}(s, p)X(s, p)\prod_{\ell \neq j}X(r_{\ell}) \nonumber\\
& & &\quad-\theta_{i}^2\prod_{j=1}^{M_R}X(r_j) \leq 0 \label{opt: stability}
\end{alignat}
\end{subequations}
where $X(s,p)=\sum_{j=1}^{M_R}X(r, s, p)$, $X(r)=\sum_{s, p}X(r, s, p)$\footnote{$X(A)$ denote the occupation~(probability) measure \cite{altman1999constrained} that assigns a probability to the event $A$}, $\theta_{i}=\overline{\lambda}\big/ \lambda_{i}$, and $P_{ij}(s, p)$ is the transition probability of the SD-MC model in \eqref{sys:ss-MC}. 
Then the optimal stationary power policy $\mu_{p}^{\star}=\{\mathbb{P}(p | r)\}_{p \in \Omega_{p}, r \in \mathcal{R}}$ and optimal probability distribution  $\pi^{\star}=\{\mathbb{P}(s)\}_{s \in S}$ for the MDP states can be represented as
\begin{align}
\label{eq:power-policy}
\mathbb{P}(p | r)&=\frac{\sum_{s \in S}X^{*}(s, r, p)}{\sum_{p \in \Omega_{p}, s \in S}X^{*}(s, r, p)} \\
\mathbb{P}(s)&=\sum_{p \in \Omega_{p}, r \in \mathcal{R}}X^{*}(s, r, p)
\label{eq:policy-s}
\end{align}
with $\{ X^{*}(s, r, p)\}_{s \in S, r \in \mathcal{R}, p \in \Omega_{p}}$ is the solution of  CPOP \eqref{opt: qcqp}.  
\end{theorem}
\begin{proof}
Please refer to the Appendix \ref{appendix} for the proof. 
\end{proof}
%\begin{remark}
%It is evident from \eqref{eq:power-policy}  that the optimal power policy is a result of  jointly considering the data rates $r$ from the communication channel side and the states $s$ from the MDP control system. 
%\end{remark}
\begin{remark}
The inequality \eqref{opt: stability} is equivalent to the stability condition in \eqref{ineq:condition-lambda} if stationary policies are considered in the system. 
For all $i$ with  $1 \leq i \leq M_{R}$, let the following function $h_i$ be those that form the inequality in \eqref{opt: stability}
$$h_{i}(X):=\sum_{j=1}^{M_R}\sum_{s, p}P_{ij}(s, p)X(s, p)\prod_{\ell \neq j}X(r_{\ell})-\theta_{i}^2\prod_{j=1}^{M_R}X(r_j)$$
It is clear that functions $\{h_{i}\}_{1 \leq i \leq M_{R}}$ are polynomial functions whose order is determined by the number  $M_{R}$ of data rates specified in the channel. 
The solutions to the CPOP \eqref{opt: qcqp} thus represent the occupation measures of the data rates, MDP states, and transmission power. 
These measures are used to construct the optimal power policy that is conditioned on the data rate $r$. 
One can refine the power policy using the optimal occupation measures $\{X^{*}(s, r, p)\}$ by defining it to be conditional on both data rate $r$ and the environment state $s$. 
Specifically, $\mathbb{P}(p~|~r, s)=X^{*}(s, r, p)/\sum_{p \in \Omega_{p}, s \in S}X^{*}(s, r, p)$ where the power policy needs to use both data rate and MDP state to determine the probability of choosing certain power level. 
In cases where environment state $s$ may not be available to the communication system, then the power policy $\mathbb{P}(p~|~r)$ is practically more feasible to implement than $\mathbb{P}(p~|~r, s)$.
\end{remark}

Not that it is well known that it is in general computationally hard~(NP hard) to even decide whether a polynomial with degree equal to or greater than three is convex over a compact region \cite{ahmadi2019complexity}. 
Such a negative result on CPOP decidability suggests that the structure of the polynomial constraints as formed in \eqref{opt: qcqp} must be investigated to have efficient solutions. 
In Subsection \ref{subsec:qcp}, we show that the CPOP \eqref{opt: qcqp} can be reduced to a quadratic constraint program for the case of \emph{two-state state-dependent Markov channel}~(two-state SD-MC). 
The two-state SD-MC can be considered as a generalization of the well known bursty erasure channel \cite{minero2012stabilization}. 
Regarding the general case of SD-MC channels,  we propose a linear programming method in Subsection \ref{subsec:lpr} to approximate the solutions of the CPOP \eqref{opt: qcqp}, which therefore allows a computationally efficient way to solve the co-design problem in industrial NCS.  

\subsubsection{Two-state SD-MC: Quadratic Constraint Programs}
\label{subsec:qcp}
This section considers a two-state SD-MC with data rates of $0$ and $R$ as depicted in Fig.~\ref{fig: two-state-MC}.
%The polynomial constraint optimization problem can be reduced to a quadratic constraint program if a two-state state-dependent Markov channel is considered. 
%The two-state SD-MC generalizes the traditional bursty packet erasure channel \cite{minero2012stabilization} by defining transition probabilities that depend on the transmission power and MDP states of the external environment. 
Under this particular two-state SD-MC model, Lemma \ref{lemma:qcp} below shows that the polynomial constraints in \eqref{opt: stability} are reduced to quadratic constraints.
In this case, the CPOP \eqref{opt: qcqp} reduces to a quadratic optimization problem which can be solved by SDP programming methods if the matrices associated with the corresponding quadratic constraints are positive semidefinite. 

\begin{lemma}
	\label{lemma:qcp}
Consider a two-state SD-MC with the data rates of $r_1$ and $r_2$ in Fig.~\ref{fig: two-state-MC}.
Given the set of MDP states $S$ and transmission power set $\Omega_{p}$, the polynomial constraints in \eqref{opt: stability} are quadratic constraints which can be formulated as
\begin{align}
\label{ineq:qc-1}
X^{T}\overline{Q}_{1}X \leq 0 \\
X^{T}\overline{Q}_{2}X \leq 0 
\label{ineq:qc-2}
\end{align}
where $X=[X(r_1, s_1, p_1), X(r_1, s_1, p_2), \ldots, X(r_1, s_2, p_1), \ldots, \allowbreak X(r_2, s_{M_s}, p_{M_p})]^{T}$ is an augmented vector whose elements are arranged in the order of transmission power $p$, MDP state $s$, and data rate $r$. 
In addition, the matrices $\overline{Q}_{1}, \overline{Q}_{2} \in \mathbb{R}^{2M_{s}M_{p} \times 2M_{s}M_{p}}$ are of the forms
\begin{align*}
\overline{Q}_{1}&=A_{11}^{T}I_{1}+A_{12}I_{2}-\theta_{1}^{2}I_{2}^{T}I_{1} \\
\overline{Q}_{2}&=A_{21}^{T}I_{1}+A_{22}I_{2}-\theta_{2}^{2}I_{2}^{T}I_{1}
\end{align*}
where $A_{ij}=[\vec{P}_{ij}, \vec{P}_{ij}]_{1 \times 2M_{s}M_{p}}$ with $\vec{P}_{ij}=Vec(P_{ij}), \forall i, j \in \{1, 2\}$, $\theta_{i}=\overline{\lambda}\big/ \lambda_{i}$ with $i=1, 2$, $I_{1}=[\mathbf{0}, \mathbf{e}]$, and $I_{2}=[\mathbf{e}, \mathbf{0}]$, in which $\mathbf{0}$ and $\mathbf{e}$ are row vectors of $M_{s}$ and $M_{p}$ of zeros and ones, respectively. 
\end{lemma}
\begin{proof}
	The proof is provided in Appendix \ref{appendix}.
\end{proof}
\begin{figure}
	\centerline{\includegraphics[width=.8\columnwidth]{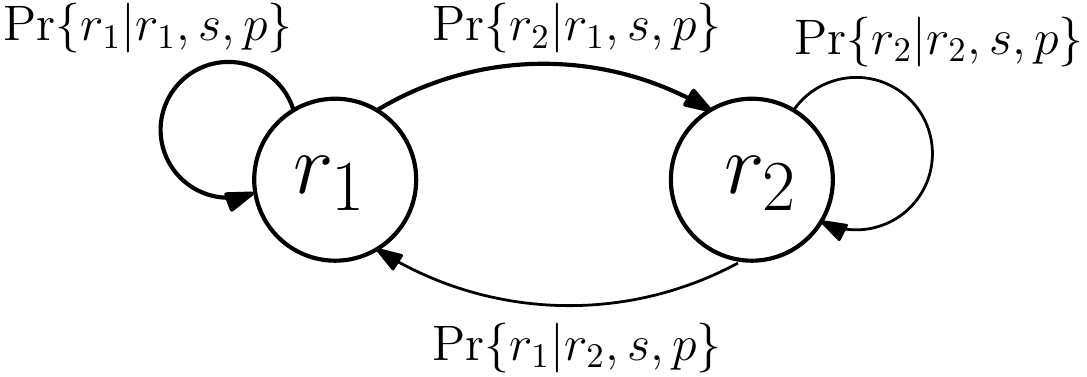}}
	\caption{A two-state SD-MC model.}
	\label{fig: two-state-MC}
\end{figure}

Based on Lemma \ref{lemma:qcp}, Theorem \ref{thm:qp-sdp} below shows that the resulting constrained quadratic optimization problem can be solved using SDP methods.  
\begin{theorem}
	\label{thm:qp-sdp}
Consider the two-state SD-MC in Fig.~\ref{fig: two-state-MC} and suppose that the condition $\overline{Q}_{1}, \overline{Q}_{2} \succeq 0$ holds for the matrices in the quadratic constraints \eqref{ineq:qc-1} and \eqref{ineq:qc-2}.
Then, the optimization problem in \eqref{opt: qcqp} can be solved using SDP methods. 
\end{theorem}
\begin{proof}
The proof is straightforward from the fact that sets formed by the constraints \eqref{ineq:qc-1} and \eqref{ineq:qc-2} are convex if matrices $\overline{Q}_{1}, \overline{Q}_{2}$ are positive semidefinite~(i.e. $\overline{Q}_{1}, \overline{Q}_{2} \succeq 0$) \cite{blekherman2012semidefinite}.
\end{proof}
\subsubsection{Linear Programming Relaxations}
\label{subsec:lpr}
This subsection discusses a linear programming (LP) relaxation of the CPOP \eqref{opt: qcqp}. 
The key idea of such a relaxation is the use of a more conservative $1$-norm form of the stability condition derived in \eqref{ineq:condition-lambda}. 
Specifically,  Theorem \ref{thm:lpr} below shows that both the ASE and ASAS properties can be attained under a conservative $T_{MATI}$ that is induced by the $1$-norm condition. Furthermore, the subsequent Proposition \ref{prop:lpr} proves that such a more conservative $1$-norm stability condition leads to linear constraints in the formulated optimization problem. 

\begin{theorem}
	\label{thm:lpr}
Consider the stochastic hybrid system in \eqref{sys:continuous}-\eqref{sys:jump}, the SD-MC model in \eqref{sys:ss-MC}, and the controlled external environment  $\mathcal{M}_{env}$. Suppose that the $T_{MATI}$ is defined by the condition in \eqref{ineq:MATI} for some constant $\overline{\lambda} \in \mathbb{R}_{+}$ that satisfies 
\begin{align}
\label{ineq:norm-1}
\overline{\lambda} > \sqrt{M_R\|\diag(\lambda_{i}^{2})\overline{P}(\mu)\|_{1}}
\end{align}
Then the stochastic hybrid system \eqref{sys:continuous}-\eqref{sys:jump} is ASE if Assumption \ref{assumption} holds.
Furthermore,  system \eqref{sys:continuous}-\eqref{sys:jump} is ASAS if the condition in \eqref{ineq:new} holds. 
\end{theorem}
\begin{proof}
	The proof is straightforward by realizing the conditions between matrix $\infty$-norm and $1$-norm. 
	Specifically, since $\|\diag(\lambda_{i}^{2})\overline{P}(\mu)\| \leq M_R \|\diag(\lambda_{i}^{2})\overline{P}(\mu)\|_{1}$ due to the matrix condition $\|A\| \leq n\|A\|_{1}, \forall A \in \mathbb{R}^{m \times n}$, then $\forall \overline{\lambda} > 0$ we have
	\begin{align*}
		\overline{\lambda} > \sqrt{M_R\|\diag(\lambda_{i}^{2})\overline{P}(\mu)\|_{1}} \geq \sqrt{\|\diag(\lambda_{i}^{2})\overline{P}(\mu)\|}
	\end{align*}
The above inequalities imply that the satisfaction of the $1$-norm condition in \eqref{ineq:norm-1} leads to the stability condition in \eqref{ineq:condition-lambda} that are shown to ensure ASE as proven in Thereom \ref{thm:ase} and ASAS as proven in Theorem \ref{thm:asas}, respectively. Thus, the proof is complete. 
\end{proof}
\begin{remark}
The use of the stability constraints \eqref{ineq:norm-1} in the co-design problem essentially trades the optimality with computational efficiency. 
Regarding the system optimality, the constrained optimization problem with constraints \eqref{ineq:norm-1} leads to sub-optimal solutions because the stability constraints in \eqref{ineq:norm-1} are more conservative than the polynomial constraints in \eqref{ineq:condition-lambda}. 
Specifically, the size of the sets formed by the constraints in \eqref{ineq:norm-1} is inversely proportional to the number of the data rates $M_{R}$. 
Despite this sub-optimality, the use of constraints \eqref{ineq:norm-1} in the co-design problem reduces the original NP-hard CPOP into more computationally tractable LP problems  which can be solved efficiently.
This fact is stated in Proposition \ref{prop:lpr} below.
\end{remark}

\begin{proposition}
	\label{prop:lpr}
	For a given $\overline{\lambda} \in \mathbb{R}_{+}$, let $\mathcal{H}_{\mu}$ denotes a set of joint policies $\mu=(\mu_{m}, \mu_{p})$ which satisfies the $\infty$-norm stability condition in \eqref{ineq:condition-lambda}, and let $\mathcal{H}_{\tilde{\mu}}$ denotes a set of joint policies $\tilde{\mu}=(\tilde{\mu}_{m}, \tilde{\mu}_{p})$ which satisfies the $1$-norm stability condition in \eqref{ineq:norm-1}. Then
	\begin{enumerate}
		\item $\mathcal{H}_{\tilde{\mu}} \subset \mathcal{H}_{\mu}$.
		\item with the $1$-norm stability condition \eqref{ineq:norm-1}, the polynomial inequalities \eqref{opt: stability} of the CPOP in \eqref{opt: qcqp} can be reduced to linear inequalities.
		\item let $\mu^{\star}$ and $J^{\star}$ denote the optimal joint policy and cost, respectively, of the optimization problem \eqref{opt: game-1} with the $\infty$-norm stability condition. Let $\tilde{\mu}^{\star}$ and $\tilde{J}^{\star}$ denote the optimal joint policy and cost, respectively, of the optimization problem \eqref{opt: game-1} with the $1$-norm stability condition. Then $J^{\star} \leq \tilde{J}^{\star}$. 
	\end{enumerate}
\end{proposition}
\begin{proof}
	The proof is in included in Appendix \ref{appendix}.
\end{proof}

%Theorem \ref{thm:qcqp} shows that the optimal transmission power policy $\mu_{p}^{\star}=\{\mathbb{P}(p | r)\}_{p \in \Omega_{p}, r \in \mathcal{R}}$ and optimal stationary distribution for the MDP states $\pi^{\star}=\{\mathbb{P}(s)\}_{s \in S}$ can be obtained by solving the polynomial constraint optimization problem. Given the solutions, one still needs to find a control policy $\mu_{m}=\{\mathbb{P}(a | s) \}_{a \in A(s), s \in S}$ to achieve the optimal stationary distribution $\pi^{\star}$. To ensure that there exists a control policy under which the induced Markov chain leads to any optimal stationary distribution $\pi^{\star}$ obtained from \eqref{opt: qcqp}, this paper assumes that the MDP used to model the external environment is an \emph{ergodic Markov process}\cite{puterman2014markov}, i.e., any state from MDP can be eventually reachable from any other state by following a suitable policy. With the ergodic assumption, the next step is to find a control policy $\mu_{m}$ to achieve the optimal stationary distribution $\pi^{\star}$ while minimizing some pre-defined costs related to control performance. Specifically, let $\{c_{m}(s, a)\}_{s \in S, a \in A}$ denote the cost function for each state-action pair of the MDP process. The optimal control policy $\mu_{m}^{\star}$ is obtained by solving the following optimization problem 
Assume that the optimal power policy $\mu_{p}^{\star}=\{\mathbb{P}(p | r)\}_{p \in \Omega_{p}, r \in \mathcal{R}}$ and optimal stationary distribution for the MDP states $\pi^{\star}=\{\mathbb{P}(s)\}_{s \in S}$ are obtained from Theorem \ref{thm:qcqp}.
The next step is then to find a control policy $\mu_{m}=\{\mathbb{P}(a | s) \}_{a \in A(s), s \in S}$ for the MDP to achieve the optimal stationary distribution $\pi^{\star}$. 
Let $\{c_{m}(s, a)\}_{s \in S, a \in A}$ denote the cost function for each state-action pair of the MDP process. 
Then the optimal control policy $\mu_{m}^{\star}$ can be obtained by solving the optimization in Problem \ref{problem-2} below.
\begin{problem}
	\label{problem-2}
Consider an ergodic MDP process $\mathcal{M}_{env}=\{S, s_{0}, A, Q\}$ with a desired stationary distribution $\pi^{\star}$ that is obtained by solving the optimization problem \eqref{opt: qcqp}.
Based on such a $\pi^{\star}$, find an optimal control policy $\mu_{m}^{\star}$ that solves the following optimization problem while attaining the desired stationary distribution $\pi^{\star}$.
	\begin{subequations}
\label{problem-control}
\begin{alignat}{2}
&\!\min_{\mu_{m}}
& & \lim_{T \rightarrow \infty}\frac{1}{T} \mathbb{E}_{s_{0}}^{\mu_{m}} \sum_{i=0}^{T} c_{m}(s_{k}, a_{k}) \label{problem-opt}\\
& \text{s.t.}
& & Q(\mu_{m})\pi^{\star}=\pi^{\star}  \label{problem-constraint}
\end{alignat}
\end{subequations}
where $Q(\mu_{m})$ is the transition matrix of the induced Markov chain under the control policy $\mu$. 
\end{problem}
%\begin{remark}
%The objective function \eqref{problem-opt} uses the average expected costs to characterize the control performance under the constraint \eqref{problem-control} that the desired stationary distribution $\pi^{\star}$ is achieved. 
%\end{remark}

Theorem \ref{thm:control-lp} below shows that Problem \ref{problem-2} can be formulated and efficiently solved using LP methods. 
\begin{theorem}
\label{thm:control-lp}
Consider an ergodic MDP process $\mathcal{M}_{env}=\{S, s_{0}, A, Q\}$ with the associated cost function $\{c_{m}(s, a)\}_{s \in S, a \in A}$.
For a given stationary distribution $\pi^{\star}=[\pi^{\star}(s_1), \ldots, \pi^{\star}(s_{M_s})]^{T}$ with $\pi^{\star}(s)$ being the probability distribution of the state $s \in S$, let $\{Y(s, a)\}_{s \in S, a \in A}$ be the decision variables in the following LP problem:
	\begin{subequations}
	\label{opt-control}
	\begin{alignat}{5}
	&\!\min_{\{Y(s, a)\}} & & \sum_{s \in S, a \in A} c_{m}(s,a)Y(s,a) \label{control-obj}\\
	& \text{s.t.}
	& & \sum_{a} Y(s, a)-\sum_{s', a} q(s |s', a)Y(s', a)=0, \forall s \in S \label{control-mdp2} \\
	& & &\sum_{s \in S, a \in A}Y(s, a)=1, \quad Y(s, a) \geq 0, \forall s, a \label{control-occupation} \\
	& & &\sum_{a \in A}Y(s, a) = \pi^{\star}(s), \forall s \in S. \label{control-stationary}
	\end{alignat}
\end{subequations}
Then, Problem \ref{problem-2} can be solved based on the LP problem formulated in \eqref{opt-control} and the corresponding optimal control policy $\mu^{\star}_{m}$ can be obtained by 
\begin{align}
\mathbb{P}(a | s)=\frac{Y^{*}(s, a)}{\sum_{a \in A}Y^{*}(s, a)}, \quad \forall s \in S, a \in A
\end{align}
where $\{Y^{*}(s, a)\}$ is the solution to the LP problem in \eqref{opt-control}.
\end{theorem}
\begin{proof}
	The proof follows straightforwardly from the LP representation of constrained MDP~(cf. equation (4.3) in \cite[Chapter 4 ]{altman1999constrained}). The equation \eqref{control-stationary} is equivalent to the constraint of the stationary distribution imposed in \eqref{problem-constraint}. 
\end{proof}
%& & \sum_{a \in A(s)} X(s, a)-\sum_{\substack{a \in A(s) \\ s' \in S}}q(s \vert s', a)X(s', a), \forall s \in S \label{opt: mdp1} \\
\section{Simulation Results}
\label{sec:simulation}
This section presents simulation results to verify the given results on stochastic stability~(Theorem \ref{thm:asas}).
The simulation results are also used to demonstrate the advantages of the proposed optimal co-design policy~(i.e. Theorem \ref{thm:qp-sdp}-\ref{thm:lpr}) by comparing the resulting system performance against that based on the separation design method which is widely adopted in current literature.

\begin{table}[!b]
	\caption{MDP $\mathcal{M}$ Transition Probability~$Q$ and Costs~$c_{m}$}
	\label{table:tp-mdp}
	\centering
	\begin{small}
	\begin{tabular}{c|c c ||c}
		& $s_1$ & $s_2$ & $c_{m}(s, a)$ \\
		\hline
		$s_1, Stay$ & $0.9$ & $0.1$ & $0.4$\\
		$s_1, Go$ & $0.1$ & $0.9$ & $0.4$\\
		\hline
		$s_2, Stay$ & $0.9$ & $0.1$ & $0.6$\\
		$s_2, Go$ & $0.1$ & $0.9$ & $0.6$\\
		\hline
	\end{tabular}
	\end{small}
\end{table}

In the simulation, a linear batch reactor process described in \cite{nesic2004input} is used for the NCS part. 
The considered state-space model of an unstable linear batch reactor process is given by
\begin{align*}
	\dot{x}=Ax+Bu
\end{align*} 
where
\begin{small}
\begin{align*}
	A&=\begin{bmatrix}
		1.38 & -0.2077 & 6.715 & -5.676 \\
		-0.5814 & -4.29 & 0 & 0.675\\
		1.067 & 4.273  & -6.654 & 5.893 \\
		0.048 & 4.273 & 1.343 & -2.104
	\end{bmatrix} \\ B &= \begin{bmatrix}
		0 & 0 \\ 
		5.679 & 0 \\
		1.136 & -3.146 \\
		1.136 & 0
	\end{bmatrix}
\end{align*}
\end{small}
The process is assumed to be controlled by a remote controller which uses an estimated value of the state $\hat{x}$ that is procuced by a model-based estimator. 
The system dynamics from the controller side may then be modeled as follows.
\begin{align*}
	\hat{x}&=A\hat{x}+Bu \\
	u&=K\hat{x}, \quad \forall t \in [t_{k}, t_{k}+T)
\end{align*}
where $T$ is a constant time internal that is selected to satisfy the MATI conditions. 
By using the emulation-based method, the state feedback controller gain $K$ which assures the exponential stability of the system without considering the network effect  is designed to be   the following.

\begin{align*}
	K=\begin{bmatrix}
		0.6961	& 0.8133 &	0.5639	& -1.8492 \\
		2.6908 & 1.1764 &	-1.2762 & 0.9968
	\end{bmatrix}
\end{align*}

\begin{table}[!b]
	\caption{$\mathcal{M}_{c}$ Transition Probability~$P_{ij}(s, p)$ and Power-Rate Costs~$c_{p}$, $c_{r}$}
	\label{table:tp-mc}
	\centering
\begin{small}
	\begin{tabular}{c|c c ||c c}
		& $r_1$ & $r_2$ & $c_{p}(p)$ & $c_{r}(r)$ \\
		\hline
		$r_1, (s_1, L)$ & $0.8$ & $0.1$ & $0.4$ & $0.6$\\
	    $r_1, (s_1, H)$ & $0.6$ & $0.4$ & $0.6$ & $0.6$\\
	    $r_1, (s_2, L)$ & $0.4$ & $0.6$ & $0.4$ & $0.4$\\
	    $r_1, (s_2, H)$ & $0.1$ & $0.9$ & $0.6$ & $0.4$\\
		\hline
		$r_2, (s_1, L)$ & $0.8$ & $0.2$ & $0.4$ & $0.6$\\
		$r_2, (s_1, H)$ & $0.6$ & $0.4$ & $0.6$ & $0.6$\\
		$r_2, (s_2, L)$ & $0.5$ & $0.5$ & $0.4$ & $0.4$\\
		$r_2, (s_2, H)$ & $0.1$ & $0.9$ & $0.6$ & $0.4$\\
	\end{tabular}
\end{small}
\end{table}

In the simulation, the controlled external environment is assumed to be a moving vehicle that is modeled as a two-state MDP with a state set of $\{s_{1}, s_{2}\}$ and an action set of $\{\text{Go}, \text{Stay}\}$. 
The MDP states represent the partitions of regions in the factory floor.  
Let the state $s_1$ denotes the region that can cause shadow fading when occupied by the vehicle, while $s_{2}$ denotes the non-shadow fading region. 
The transition probabilities of the states under each action are provided in Table \ref{table:tp-mdp}. 
Under the external environment, a two-state SD-MC model is used to simulate the \emph{state-dependent} fading channel that is being utilized by the networked batch reactor process. 
In the two-state SD-MC model, the states of $r_{1}=0$ and $r_{2}=2$ are selected to represent the data rates supported by the wireless communication system. 
A set $\{H, L\}$ is also selected for transmission power with $H$ and $L$ represent the high and low transmission power, respectively. 
The state-dependent transition probabilities under different environment states and transmission power are summarized in Table \ref{table:tp-mc}.

Our simulations consider scenarios where the interests of the MDP system are in conflict with the objective of the NCS.
Specifically, the shadowing state $s_1$ is favored by the MDP system and is associated with a lower cost than the un-shadowing state $s_2$~($c_{m}(s_1, a)=0.4$ and $c_{m}(s_2, a)=0.6$ as shown in Table \ref{table:tp-mdp}). Such scenarios represent challenging situations that are commonly encountered in complex industrial environments where cooperative and adaptive strategies are needed the most.

\subsection{Stochastic Stability}
This subsection focuses on the validation of the stochastic stability result as proved in Theorem \ref{thm:asas} as well as the investigation of how tight the sufficient conditions on MATI are by comparing them against the necessary bounds generated by Monte Carlo simulations. 

\begin{figure}[!t]
	\centerline{\includegraphics[width=.9\columnwidth]{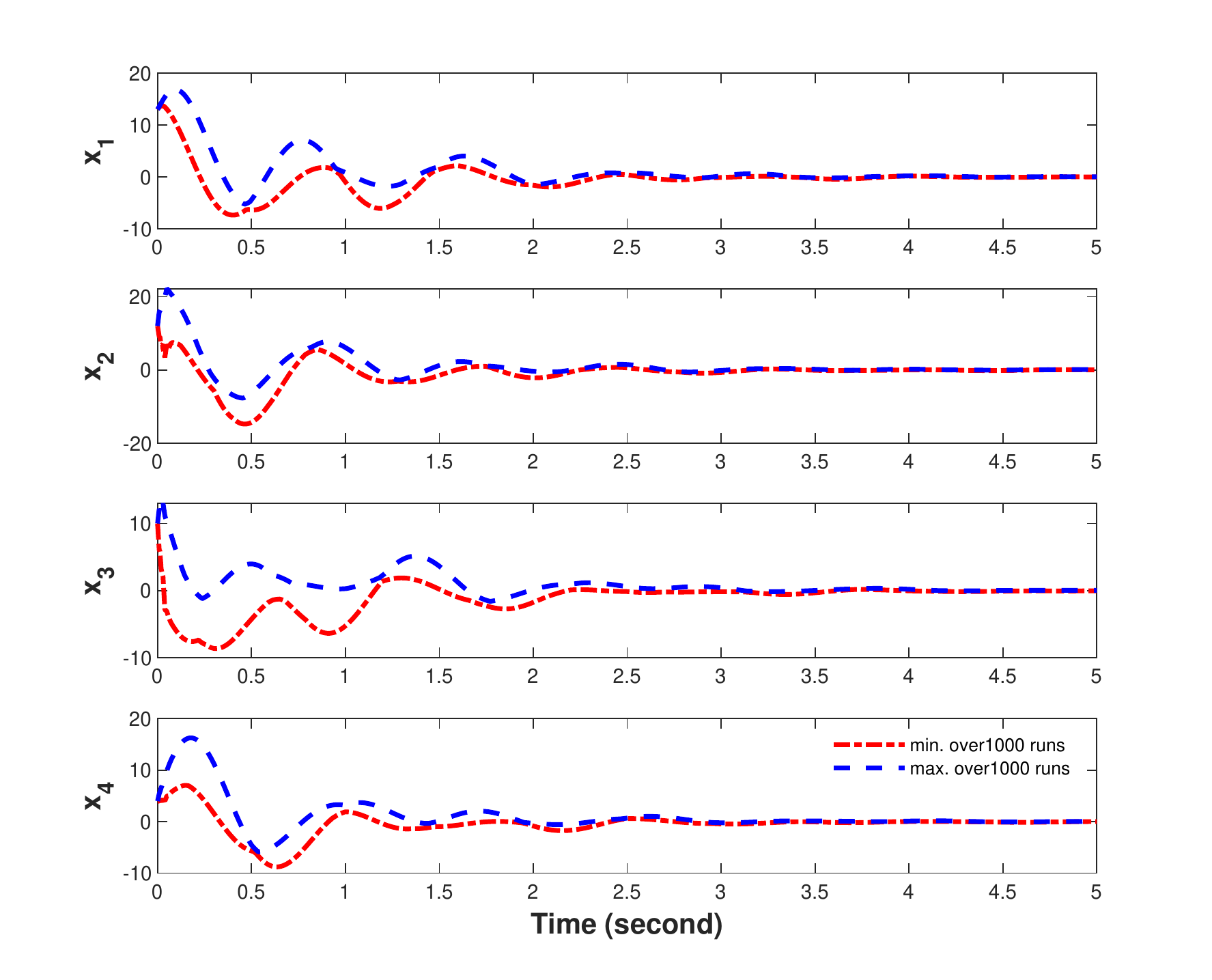}}
	\caption{The maximum and minimum values of the batch reactor process state trajectories for a transmission interval $T=0.01 s$.}
	\label{fig: ss}
\end{figure}

Given the considered process and channel models, the dynamics of the aggregated error states $\overline{e}=[e;\xi]$ where the time derivatives of $e=x-\hat{x}$ and the $\xi$ can be written as $\dot{e}=Ae$ and $\dot{\xi}=|A|\xi$, respectively, then the function $W(\overline{e})$ in Assumption \ref{assumption} can be selected as $W(\overline{e})=|\overline{e}|^2$. With the selected functions and parameters for the batch reactor process, one can determine the parameters $L=17.8870$ and $\zeta=26.5415$. 
As shown in inequality \eqref{ineq:condition-lambda}, the parameter $\overline{\lambda}$ depends on the characteristics of the SD-MC channel. 
With the data rates of $r_{1}=0$ and $r_{2}=2$, then $\lambda_{0}=1$ and $\lambda_{1}=0.5$. 
Let $\overline{P}(\mu)=[0.2, 0.2;0.8, 0.8]$ denotes a transition matrix under a selected control and transmission power policy $\mu$. 
By inequality \eqref{ineq:condition-lambda}, then $\overline{\lambda} > \sqrt{\|\diag(\lambda_{i}^{2})\overline{P}(\mu)\|}=0.6325$. 
Let $\overline{\lambda}=0.64$ and $MATI = 0.0104 s$ so that the transmission time interval is selected to be $T=0.01 s \leq MATI$. 
Fig.~\ref{fig: ss} shows the maximum~(blue dashed line) and minimum~(red dash-dot line) values of the system states evaluated over $1000$ runs with similar initial values and $T=0.01 s$. 
These plots show that both the maximum and minimum trajectories asymptotically converge to the origin, implying that the ASAS property stated in Theorem \ref{thm:asas} under the MATI condition in \eqref{ineq:MATI} holds.

\begin{figure}[!t]
	\centerline{\includegraphics[width=.8\columnwidth]{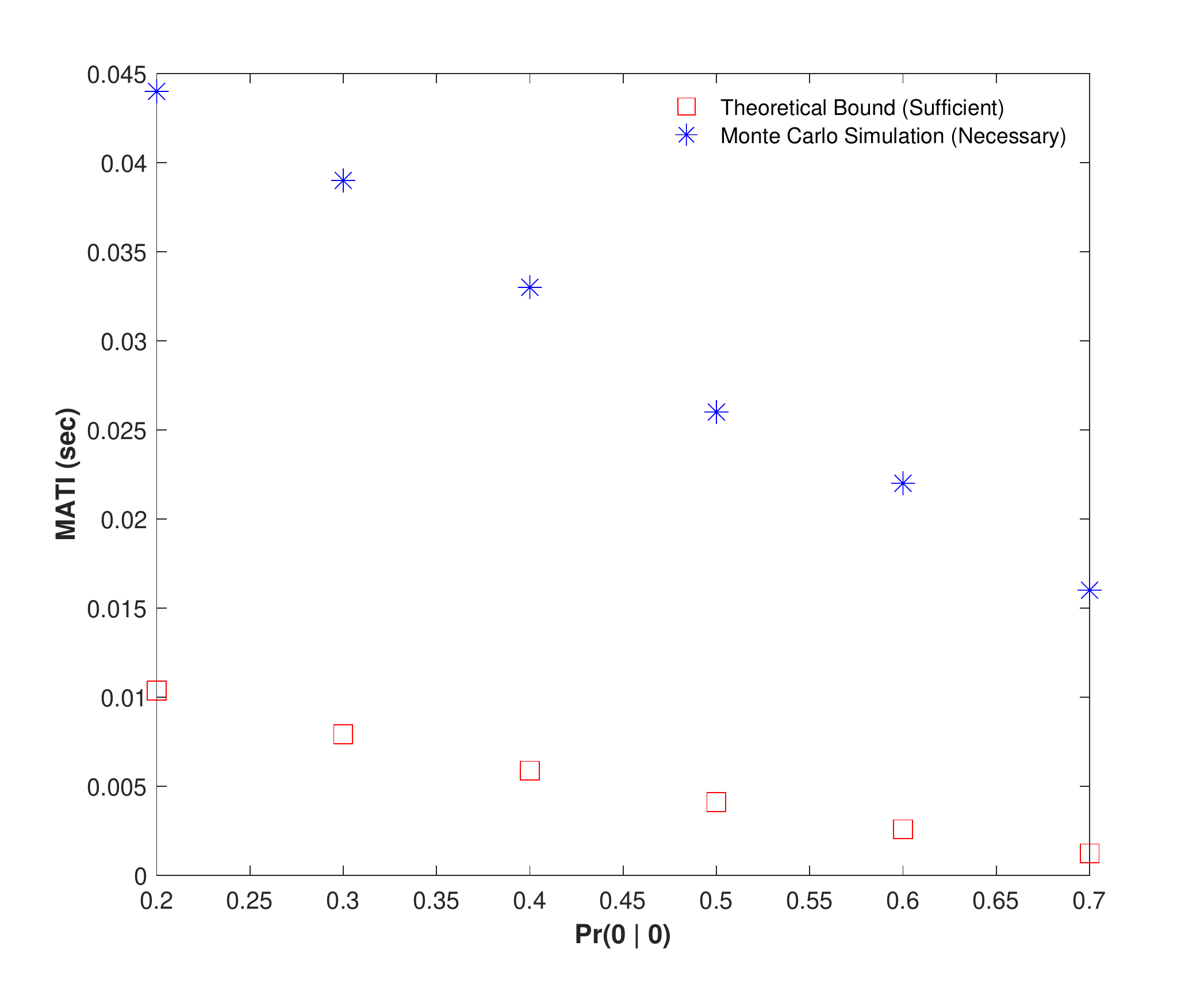}}
	\caption{Comparison of sufficient and necessary bounds on MATI under different channel conditions $\mathbb{P}\{r'= 0 |r=0\}=0.2, 0.3, 0.4, 0.5, 0.6, 0.7$ with $\mathbb{P}\{r'=0 | r=2 \} = 0.2$.}
	\label{fig: MATI}
\end{figure}

In addition, we are also interested in investigating how tight the derived MATI bound is by comparing it against the necessary bounds. 
In Fig.~\ref{fig: MATI}, the theoretical MATI bounds~(marked by red square) derived in inequality \eqref{ineq:MATI} under a variety of channel conditions specified by different transition probabilities of $\mathbb{P}(0~|~0)$ in the Markov channel,  are compared against the necessary bounds~(marked by blue star) that are generated by Monte Carlo simulations. 
Specifically, the necessary MATIs are obtained by gradually increasing the transmission interval $T$ until the system is \emph{almost surely asymptotically unstable}. 
In the simulation, the system is asserted to be almost surely asymptotically unstable if both maximum and minimum values of the system trajectories evaluated over $1000$ samples go unbounded. 
The comparison results show that the MATI bounds in \eqref{ineq:MATI} are around 4 to 8 times more conservative than the necessary bounds generated by Monte Carlo simulations. 
It is worth noting that the performance gap is reasonably close when considering that similar conservativeness were also reported in \cite{nesic2004input} for deterministic system cases. 
Moreover, there exists a trade-off between the MATI bound for stochastic stability and the channel condition that is dictated by the transition probabilities of the SD-MC channel. 
As shown by the simulation results in Fig.~\ref{fig: MATI}, both sufficient and necessary MATI bounds decrease as the channel conditions are degraded\footnote{The increasing value of the transition probability of staying in bad channel state, i.e., $\mathbb{P}(0 | 0)$, is used to simulate the degradation of channel conditions.}. 
The adaptive nature of the MATI as shown by the simulation results is due to the fact that the transmission frequency must be increased to compensate for the packet losses caused by bad channel conditions in order to ensure stochastic stability. 
Since transition probabilities of the SD-MC channel can be controlled through the selection of different control and transmission power policies,  such state-dependent features are exploited to obtain jointly control and communication policies that can attain optimal system performance while assuring stochastic stability. 
%\subsubsection{Nonlinear DC Motor System}

\subsection{Optimal Co-design Policy: Quadratic and Linear Programs}
This subsection demonstrates the advantages of using optimal co-design strategies as proposed in Theorem \ref{thm:qcqp}-\ref{thm:lpr} by comparing them against traditional separation design methods. 

In traditional separation design methods (e.g. \cite{gatsis2014optimal}), the design of the optimal power policy is assumed to be independent of the control policy design. 
To ensure a fair comparison, this simulation uses similar parameters for both the co-design and the separation design methods.
The comparison results are also evaluated for a wide range of channel conditions. 

In the simulation, various shadow fading levels are simulated by selecting different transition probability $\mathbb{P}(0|0, s_{1}, H)$ under the shadow fading state $s_1$ and high transmission power $H$ in the channel model. 
For the separation design method, the optimal power policy is designed to minimize only the communication costs\footnote{$\min_{\mu_{p}}\lim_{\ell \rightarrow +\infty}\frac{1}{\ell}\mathbb{E}\sum_{k=0}^{\ell}[c_{p}(p_k)+c_{r}(R_k)]$ with $c_{p}, c_{r}$ listed in Table \ref{table:tp-mc}} while respecting the stability constraint. 
The optimal control policy is determined to minimize the costs\footnote{$\min_{\mu_{m}}\lim_{\ell \rightarrow +\infty}\frac{1}{\ell}\mathbb{E}\sum_{k=0}^{\ell}[c_{m}(s_k, a_k)]$ with $c_{m}$ defined in Table \ref{table:tp-mdp}} induced when controlling the vehicle movements. 
Fig.~\ref{fig: comp-codesign-separation} shows the comparison results of optimal joint costs generated by the separation design method~(marked by a black solid-dotted line) and the co-design strategy~(the quadratic programming result is marked by a red dotted line, while the LP result is marked by a blue dash line). 
As shown in the plots, the quadratic programming-based co-design method leads to lowest costs across the whole range of the shadow fading among all three strategies. 
Interestingly, co-design strategies under both quadratic and linear programmings are more robust in the high shadow fading regime~(i.e., the region between $0.3$ and $0.55$) than that of  the separation design in the sense that the optimal cost curves under the former two are relatively flat regardless of the fading levels, while that based on the separation design method  linearly increases as the fading level is increased.
%In traditional separation design methods, such as \cite{gatsis2014optimal}, the design of optimal power policy is assumed to be independent of the control policy design. To make a fair comparison, this simulation uses the same parameters for both co-design and separation design methods, and the comparison results are evaluated under a wide range of channel conditions. In the simulation, different shadow fading levels can be simulated by selecting a group of transition probabilities $\mathbb{P}(0~|~0, s_{1}, H)$ with values ranging from $0$ to $0.7$. The selection of the transition probability $\mathbb{P}(0~|~0, s_{1}, H)$ is more representative of the shadow fading effect than other transition probabilities in that its value reflects the likelihood of channel being in bad conditions under the shadowing state $s_{1}$ and the selection of a high transmission power $H$. The other transition probabilities are listed in Table \ref{table:tp-mc} along with the corresponding communication and control costs. 

\begin{figure}[!t]
	\centerline{\includegraphics[width=1\columnwidth]{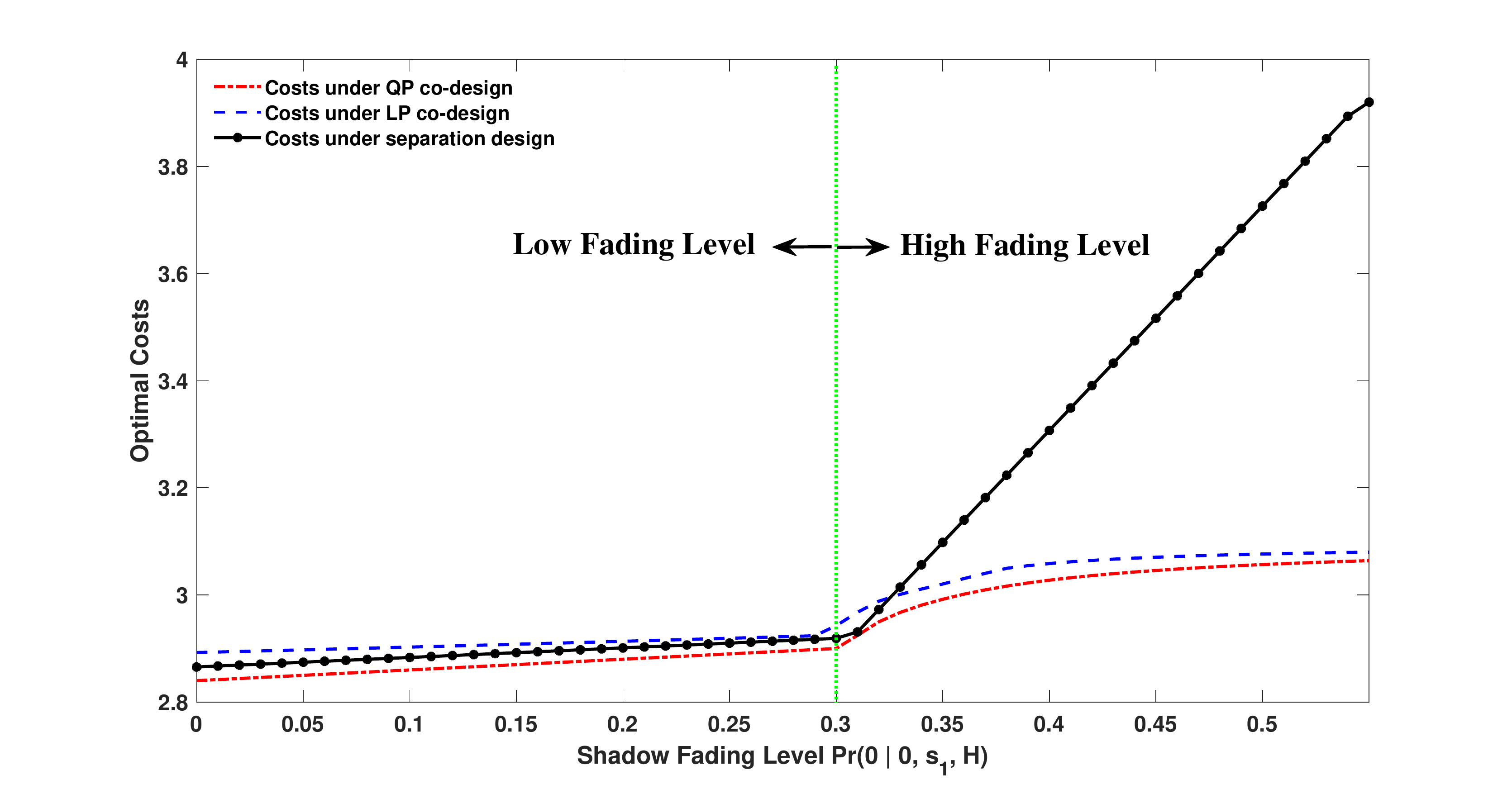}}
	\caption{Performance comparison of the proposed co-design method against the separation method under a wide range of channel conditions ranging from $0$ to $0.55$.}
	\label{fig: comp-codesign-separation}
\end{figure}

One important feature of the proposed co-design strategy which distinguishes itself from the separation design methods is its adaptive nature to different channel conditions by adjusting the joint communication and control strategies. 
To demonstrate such an adaptive nature, simulation results are depicted in Fig.~\ref{fig: optimal-costs-policy} to show the optimal costs and policies generated by the co-design strategies under the quadratic programming~(Lemma \ref{lemma:qcp}) and the LP~(Theorem \ref{thm:control-lp}) methods. 
The top plot in Fig.~\ref{fig: optimal-costs-policy} compares the optimal costs under the LP~(marked by a blue dashed line) and the quadratic programming~(marked by a red dashed-dotted line) methods for a range of shadow fading levels. 
It is expected that the quadratic programming generates lower costs than the LP one since the latter is an approximation of the former. 

\begin{figure}[!t]
	\centerline{\includegraphics[width=1\columnwidth]{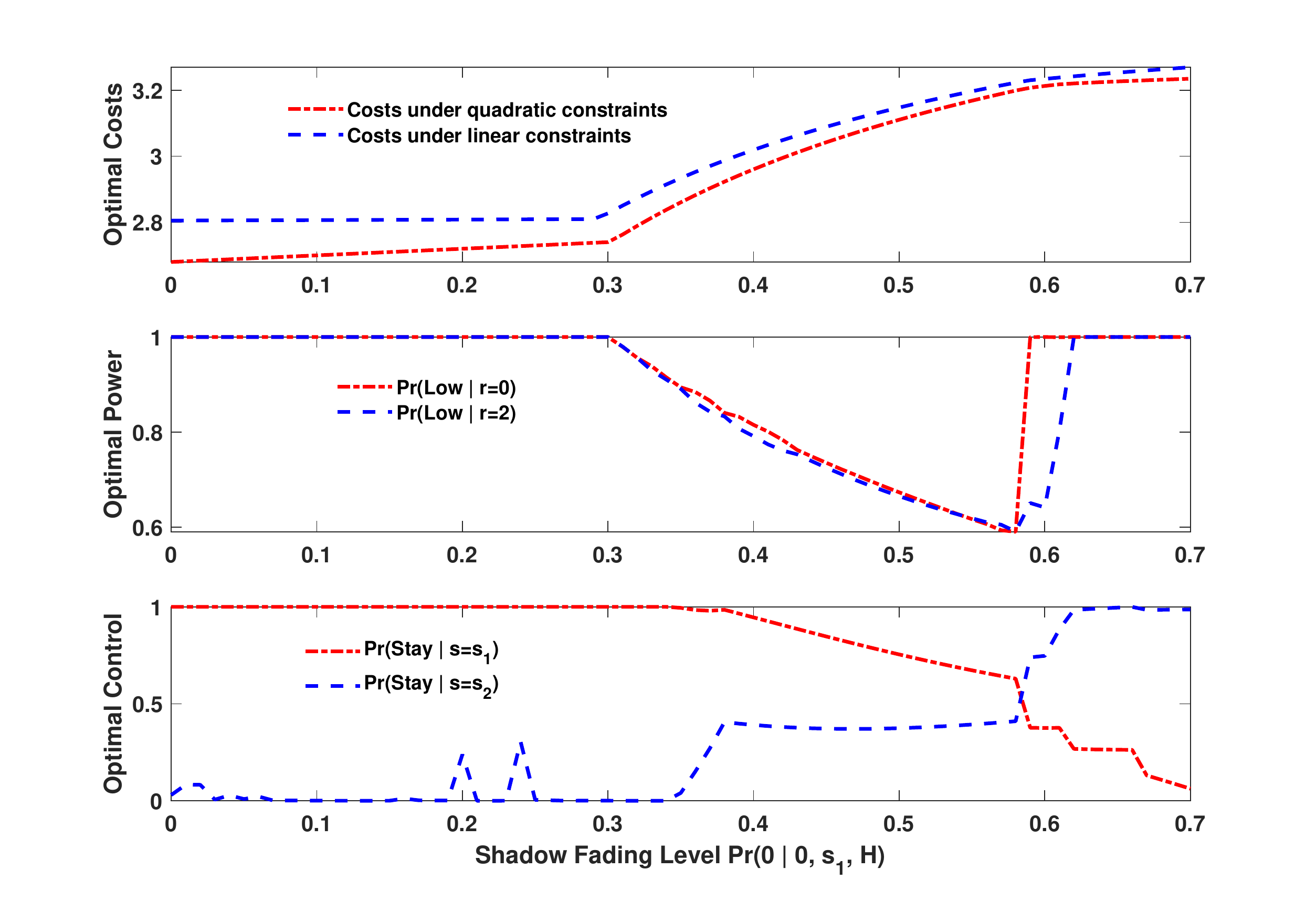}}
	\caption{Optimal costs and policies of the proposed co-design method.}
	\label{fig: optimal-costs-policy}
\end{figure}

As the shadow fading level~($\mathbb{P}(0~|~0, s_{1}, H)$) increases, the optimal costs are increased as a result of utilizing more system resources to ensure system stability under degraded channel conditions. Moreover, the co-design strategies adaptively adjust the optimal power and control policies when different fading levels occur.
The middle plots in Fig.~\ref{fig: optimal-costs-policy} show the adaptive changes of optimal power policies~($\mathbb{P}(\text{Low}~|~r=0)$ and $\mathbb{P}(\text{Low}~|~r=2)$) while the bottom plots present the optimal control policies of $\mathbb{P}(Stay~|~s_{1})$ and $\mathbb{P}(Stay~|~s_{2})$ under different fading levels.

When shadow fading is in the low level range and the value of $\mathbb{P}(0~|~0, s_{1}, H)$ ranges from $0$ to $0.3$, the optimal power policy is that of choosing the low power level with probability $1$, regardless of the data rates.
Correspondingly, the optimal control policy is to choose the "Stay" action if the current state is $s_1$ and take the "Go" action if the state is $s_{2}$~(i.e., $\mathbb{P}(Stay~|~s=s_{2}) \approx 0$ is shown as blue dashed line of the middle plots of Fig.~\ref{fig: optimal-costs-policy}). 
The underlying incentives of selecting these power and control policies lie in the fact that good channel conditions (i) allow the communication system to use a low transmission power to attain minimum use of communication resources while also (ii) enable the control system to adopt the strategy of staying in the state $s_1$, which is favored to optimize the performance of the MDP system. 
However, as the shadow fading level is increased from $0.3$ to $0.6$ and towards an extent at which the system stability constraints are about to be violated, the optimal strategies start to decrease the likelihood of using the low transmission power and push the MDP system away from the shadowing state $s_1$\footnote{The decrease on the probability of choosing the "Stay" action at state $s_1$, and the increase on the probability of choosing the "Stay" action at state $s_2$ imply that the MDP system is trying to avoid the shadowing state $s_1$ to alleviate the impact of shadow fading on the system stability}. 
Interestingly, when the fading level is approaching $0.6$, the optimal power policies start to jump back to use the low power level that saves communication resources, while the optimal control policies still push the MDP system away from the shadowing state $s_1$. 
These observations thus illustrate the adaptive and cooperative nature of the proposed co-design method which essentially adjusts the optimal control and communication policies for the benefits of whole system.

\section{Conclusions}
\label{sec:conclusion}
This paper has presented a co-design paradigm to ensure both the stability and good performance of industrial NCSs under \emph{state-dependent fading channels}. 
A novel SD-MC model was proposed to characterize the correlation between channel conditions and external environments. 
The proposed channel model was used to derive sufficient conditions on MATI that will ensure the \emph{asymptotically stable in expectation} and \emph{almost surely asymptotically stable} properties of the NCS. 
The derived stability conditions were then imposed as hard constraints in the co-design problem whose optimal solutions can be obtained from the solutions of a constrained polynomial optimization problem (CPOP) formulation.
This paper further demonstrated that such a CPOP can be either reformulated as a semidefinite programming problem (when a two-state Markovian channel is considered), or approximated as a linear programming problem with sub-optimal solution (when a more general Markovian channel is considered). 
Numerical simulation results using an industrial NCS model which consists of a batch reactor and an autonomous vehicle dynamic models were provided to demonstrate the benefits of the proposed co-design approach. 

\section{Appendices}
\label{appendix}
\begin{proof}[Proof of Theorem \ref{thm:ase}]
	Let $\overline{x}:=[x^{T}, \overline{e}, \tau, k]^{T}$ denote an augmented state, and $F(\overline{x}):=[\tilde{f}(x, e)^{T}, \overline{g}(x, \overline{e})^{T}, 1, 0]^{T}$. Define a candidate Lyapunov function for the augmented state $\overline{x}$ as below
	\begin{align}
	U(\overline{x}):=V(x)+ \zeta \phi(\tau) W^{2}(k, \overline{e})
	\end{align}
	where the function $\phi(\tau): [0, T_{MATI}] \rightarrow \mathbb{R}$ is the solution to the following nonlinear differential equation
	\begin{align}
	\dot{\phi}=-2L\phi - \zeta (\phi^{2}+1), \quad \phi(0)=\overline{\lambda}^{-1}
	\end{align}
	Following the proof of \cite[Theorem 1]{carnevale2007lyapunov}, for all $\tau$ and $k$, then
	\begin{align}
	\langle \nabla U(\overline{x}), F(\overline{x}) \rangle \leq - \varrho(|x|) - \varrho(W(k, \overline{e})) \leq -\tilde{\varrho}(\overline{x})
	\end{align}
	where $\tilde{\varrho}(\overline{x})$ is a continuous and positive definition function. It then follows that there exists a class $\mathcal{KL}$ function $\beta$ such that $\forall \lambda \in \{\lambda_{i}\}_{i=1}^{M_R}$,
	\begin{align}
	\label{ineq:U-continuous}
	U(\overline{x}(t)) \leq \beta\left(U(\overline{x}(t_{k}^{+})), t-t_{k}\right), \forall t \in [t_{k}, t_{k+1}], \forall k \in \mathbb{N}
	\end{align}
	Note that the randomness of the system comes from the stochastic jump equations in \eqref{sys:jump}, where the sequence of the augmented state $\{\overline{e}(t_{k}^{+})\}$ is a stochastic process whose probability distribution is governed by the SD-MC. 
	Let $\mathbbm{1}_{A}$ be an indicator function with value $1$ when the sample value falls in the set $A$ and value $0$ otherwise. 
	For a given set of data rates $\mathcal{R}=\{r_{i}\}_{i=1}^{M_R}$, let $U_{k+1}:=U(\overline{x}(t_{k+1}^{+}))$, and define a vector $\overline{U}_{k+1}:=\left[\mathbb{E}[U_{k+1}\mathbbm{1}_{r_{1}}], \ldots, \mathbb{E}[U_{k+1}\mathbbm{1}_{r_{i}}], \ldots, \mathbb{E}[U_{k+1}\mathbbm{1}_{r_{M_R}}]\right]^{T}$ where for all $ i, j$ we have that
	\begin{align}
	\mathbb{E}&[U_{k+1}\mathbbm{1}_{r_{i}}] \nonumber \\
	&=\mathbb{E}[V(x(t_{k+1}^{+}))\mathbbm{1}_{r_{i}}]+\zeta \phi(\tau^{+})\mathbb{E}[W^{2}(k+1, \overline{J}(k, \overline{e}, r))\mathbbm{1}_{r_{i}}] \nonumber \\
	&\overset{(a)}{=}V(x(t_{k+1})+\zeta \phi(0)\mathbb{E}[W^{2}(k+1, \overline{J}(k, \overline{e}, r))\mathbbm{1}_{r_{i}}] \nonumber \\
	&\overset{(b)}{\leq} V(x(t_{k+1})+\zeta \overline{\lambda}^{-1}\sum_{j=1}^{M_R}\mathbb{E}[\lambda_{i}^{2} W^{2}(k, \overline{e})\mathbbm{1}_{r_{j}}\mathbbm{1}_{r_{i}}] \nonumber \\
	&=V(x(t_{k+1})+\zeta \overline{\lambda}^{-1}\lambda_{i}^{2}\sum_{j=1}^{M_R}\mathbb{P}(r_{i} \vert r_{j})\mathbb{E}[ W^{2}(k, \overline{e})\mathbbm{1}_{r_{j}}] \nonumber \\
	&\overset{(c)}{=}V(x(t_{k+1})+\zeta \overline{\lambda}^{-1}\lambda_{i}^{2}\sum_{j=1}^{M_R}\underbrace{\sum_{s \in S, p \in \Omega_{p}}\mathbb{P}(r_{i} \vert r_{j}, s, p)\mathbb{P}(s, p | r_{j})}_{\overline{P}_{ij}(\mu)} \nonumber \\
	&\times \mathbb{E}[ W^{2}(k, \overline{e})\mathbbm{1}_{r_{j}}]
	\label{ineq:EU}
	\end{align}
	The equality $(a)$ in \eqref{ineq:U-continuous} holds due to the fact that the state $x$ is continuously evolved and is independent of the random jump caused by the Markov channel and MDP process. 
	The inequality $(b)$ holds because of the conditions \eqref{ineq:W} in Assumption \ref{assumption}. 
	The equality $(c)$ holds due to the used SD-MC model defined in \eqref{sys:ss-MC} and the MDP process. 
	Since \eqref{ineq:EU} holds for all $r_{i}$ with $ i=1,2,\ldots, M_R$ and $j=1,2,\ldots, M_S$, then the vector $\overline{U}_{k+1}$ satisfies 
	\begin{align}
	&\overline{U}_{k+1}\leq V(x(t_{k+1}))\mathbf{e}_{M_{R}} \nonumber \\
	&+ \zeta \overline{\lambda}^{-1}\underbrace{\begin{bmatrix}
		\lambda_{1}^{2}\overline{P}_{11} &  \lambda_{1}^{2}\overline{P}_{12} & \cdots & \lambda_{1}^{2}\overline{P}_{1M_{R}} \\
		\vdots & \vdots  &\vdots & \vdots \\
		\lambda_{i}^{2}\overline{P}_{i1}&  \lambda_{i}^{2}\overline{P}_{i2} & \cdots & \lambda_{i}^{2}\overline{P}_{iM_{R}} \\
		\vdots & \vdots & \vdots & \vdots \\
		\lambda_{M_R}^{2}\overline{P}_{M_{R}1}   & \lambda_{M_R}^{2}\overline{P}_{M_{R}2}  & \cdots & \lambda_{M_R}^{2}\overline{P}_{M_{R}M_{R}}
		\end{bmatrix}}_{\diag(\lambda_{i}^{2})\overline{P}(\mu)} \nonumber \\ 
	&\times \underbrace{\begin{bmatrix}
		\mathbb{E}[ W^{2}(k, \overline{e})\mathbbm{1}_{r_{1}}] \\
		\mathbb{E}[ W^{2}(k, \overline{e})\mathbbm{1}_{r_{2}}] \\
		\vdots \\
		\mathbb{E}[ W^{2}(k, \overline{e})\mathbbm{1}_{r_{M_{R}}}]
		\end{bmatrix}}_{\overline{W^{2}}(k, \overline{e})} \label{ineq:bar-U}
	%&=V(x(t_{k+1}))\mathbf{e}_{M_R}+ \zeta \overline{\lambda}\diag(\lambda_{i}^{2})P^{T}(\mu^{m}, \mu^{p})
	\end{align} 
	where $\mathbf{e}_{M_RS}:=[1, 1, \ldots, 1]^{T}$ is a column vector with $M_R$ of ones, and the transition matrix $\overline{P}(\mu)$ is a function of the joint policy $ \mu=(\mu_{m}, \mu_{p})$. 
	Specifically, $\forall i, j$, then 
$\overline{P}_{ij}(\mu)=\sum_{s \in S, p \in \Omega_{p}}\mathbb{P}(r_{i} \vert r_{j}, s, p)\mathbb{P}(s, p | r_{j})$ where the conditional probability $\mathbb{P}(s, p | r_{j})$ can be viewed as a joint policy which specifies the likelihood of selecting the MDP state $s \in S$ and transmission power $p \in \Omega_{p}$ when the channel date rate is $r_{j} \in \mathcal{R}$. 
By taking the infinity norm of both sides of the inequality in \eqref{ineq:bar-U}, one further has
	\begin{align}
	&|\overline{U}_{k+1}| \nonumber \\&\overset{(d)}{\leq} V(x(t_{k+1})) +\zeta \overline{\lambda}^{-1}\|\diag(\lambda_{i}^{2}\mathbb{I}_{M_S})\overline{P}^{T}(\mu^{m}, \mu^{p})\||\overline{W^{2}}(k, \overline{e})| \nonumber \\
	& \overset{(e)}{\leq} V(x(t_{k+1}))+\zeta \overline{\lambda}^{-1} \overline{\lambda}^{2}|\overline{W^{2}}(k, \overline{e})| \nonumber \\
	& = V(x(t_{k+1}))+\zeta  \overline{\lambda}|\overline{W^{2}}(k, \overline{e})|\nonumber \\
	& \overset{(f)}{\leq} |\overline{U}_{k}|
	\label{ineq: U}
	\end{align}
%\newpage
	The inequality $(d)$ in \eqref{ineq: U} holds due to the norm condition  $|y| \leq |Ax| \leq \|A\||x|$, $\forall x \in \mathbb{R}^{n}, y \in \mathbb{R}^{m}$ and $A \in \mathbb{R}^{n \times m}$. The inequality $(e)$ holds due to the condition in \eqref{ineq:condition-lambda}. 
	The inequality $(f)$ holds by the claim that $|\overline{U}(\overline{x})|=V(x)+\zeta \phi(\tau)|\overline{W^{2}}|$. 
	To prove this claim, note that the expectations of the positive functions $U(\overline{x})$ and $W(k, \overline{e})$ are positive for all $\overline{x}, \overline{e}$. 
	Since $\overline{U}(\overline{x})=V(x)\mathbf{e}+\zeta \phi(\tau)\overline{W^{2}}(k, \overline{e})$, one may proves the claim by noting that $|\overline{U}(\overline{x})|=|V(x)\mathbf{e}+\zeta \phi(\tau)\overline{W^{2}}(k, \overline{e})|=V(x)+\zeta \phi(\tau)|\overline{W^{2}}(k, \overline{e})|$, because $V(x), \zeta, \phi(\tau) \geq 0$ is constant with respect to the expectation operator. 
	The condition in \eqref{ineq: U} thus ensures that the expected value of the Lyapunov function is non-increasing at the stochastic jump states. 
	By \eqref{ineq:U-continuous}, we prove that the Lyapunov function $U(\overline{x})$ asymptotically decreases deterministically in-between the stochastic jumps. 
	Combining inequalities \eqref{ineq:U-continuous} and \eqref{ineq: U}, it is then straightforward to prove that there exists a class $\mathcal{KL}$ function $\tilde{\beta}$ such that 
	\begin{align}
	\mathbb{E}[U(\overline{x}(t))] \leq \tilde{\beta}(U(\overline{x}(0)), t)
	\end{align}
	This thus completes the proof. 
\end{proof}

\begin{proof}[Proof of Theorem \ref{thm:asas}]
	The idea of the proof is to first show that the stochastic hybrid system is \emph{exponentially stable in expectation} (ESE) under the more restrictive condition in \eqref{ineq:new}. 
	By Borel-Cantelli Lemma, it can be shown that \emph{exponential stability in expectation} implies ASAS as defined in \eqref{ineq:ASAS} in Definition \ref{def: ss}. 
	To show the ESE property, consider the Lyapunov function $U(\overline{x}):=V(x)+ \zeta \phi(\tau) W^{2}(k, \overline{e})$.
	The following inequality holds for each transmission time interval, i.e., $\forall \tau \in [0, T_{MATI}]$:
	\begin{align*}
	\langle \nabla U(\overline{x}), F(\overline{x}) \rangle &\leq - \varrho |x|^{2} - \varrho W(k, \overline{e}) 
	\leq - \varrho |x|^{2} - \varrho \overline{\alpha}_{W} |\overline{e}|^2 \\
	& \leq -(\varrho+\varrho \overline{\alpha}_{W})*(|x|^{2}+|\overline{e}|^2) = -\tilde{\varrho} |\overline{x}|^{2}
	\end{align*}
	with $\tilde{\varrho}=\varrho+\varrho \overline{\alpha}_{W}$. 
	By the conditions in \eqref{ineq:new}, it is straightforward to show that $\underline{\alpha}_{U} |\overline{x}|^{2} \leq U(\overline{x}) \leq \overline{\alpha}_{U} |\overline{x}|^{2}$ for positive constants $\underline{\alpha}_{U}=\underline{\alpha}_{V}+\zeta \overline{\lambda}	\underline{\alpha}_{W}$ and $ \overline{\alpha}_{U}=\overline{\alpha}_{V}+\zeta \overline{\lambda}^{-1} \overline{\alpha}_{W}$ where $\overline{\lambda}$ is defined in \eqref{ineq:condition-lambda}. 
	Then, one can show that $U(\overline{x}(t, k)) \leq \exp\left(-\frac{\tilde{\varrho} }{\overline{\alpha}_{U}}(t-t_k)\right) U(\overline{x}(t_k, k)), \forall k \in \mathbb{N}, t-t_k \in [0, T_{MATI}]$. Moreover, Theorem \ref{thm:ase} shows that $\mathbb{E}[U(\overline{x}(t, k+1))] \leq \mathbb{E}[U(\overline{x}(t, k))], \forall k \in \mathbb{N}$ if the condition in \eqref{ineq:condition-lambda} holds. 
	By combining the dynamics of both the continuous and discrete flows of the stochastic hybrid system as proved above, one may shows that $\mathbb{E}[U(\overline{x}(t, k))] \leq \exp\left(-\frac{\tilde{\varrho} }{\overline{\alpha}_{U}} t \right) U(\overline{x}(0, 0))$ which implies that $\mathbb{E}[|\overline{x}(t, k)|^{2}] \leq \exp\left(-\frac{\tilde{\varrho} }{\overline{\alpha}_{U}} t \right) \frac{\overline{\alpha}_{U}}{\underline{\alpha}_{U}} |\overline{x}(0,0)|^{2}$. 
	Following the proof of Theorem 9 in \cite{hu2019co}, one can shows that ESE implies ASAS by using the Borel-Cantelli Lemma. 
	Please refer to \cite{hu2019co, hu2019optimal} for more details of the proof. 
\end{proof}

\begin{proof}[Proof of Theorem \ref{thm:qcqp}]
	The proof is based on the occupation method used in \cite{altman1999constrained}. 
	For any stationary control and power policy, let $X(r, s, p):=\mathbb{P}(r, s, p)$ denotes the corresponding stationary probability distribution of the states $s, r$ and the action $p$ over the set $S \times \mathcal{R} \times \Omega_{p}$. 
	Under the stationary policy space, the objective function can be equivalently represented by \eqref{opt: obj} due to the following conditions
	\begin{align*}
	&\lim_{T \rightarrow \infty}\frac{1}{T}\mathbb{E}\sum_{i=0}^{T}c(s_k, p_k, R_k) = \lim_{T \rightarrow \infty}\frac{1}{T}\sum_{i=0}^{T} \mathbb{E}[c(s_k, p_k, R_k)] \\
	=& \lim_{T \rightarrow \infty}\frac{1}{T}\sum_{i=0}^{T} \sum_{s \in S, r \in \mathcal{R}, p \in \Omega_{p}} c(s, r, p)\mathbb{P}(s_k=s, p_k=p, R_{k}=r) 
	\end{align*}
	\begin{align*}
	=&\sum_{s \in S, r \in \mathcal{R}, p \in \Omega_{p}} \bigg[ \underbrace{\lim_{T \rightarrow \infty}\frac{1}{T}\sum_{i=0}^{T}\mathbb{P}(s_k=s, p_k=p, R_{k}=r)}_{X(s, p, r)} c(s, p, r)
	\bigg] \\
	=& \sum_{s \in S, r \in \mathcal{R}, p \in \Omega_{p}} X(s, p, r)c(s, p, r).
	\end{align*}
The conditions in \eqref{opt: mdp2} and \eqref{opt: occupation} are there to ensure the Markov property of the channel model in \eqref{sys:ss-MC} under the stationary policy as well as the probability definition of $\{X(r, s, p)\}$. 
To prove the equivalence between the constraint in \eqref{opt: stability} and the stability condition in \eqref{ineq:condition-lambda}, consider that $\forall 1 \leq i \leq M_R$:
\begin{align}
&\|\diag(\lambda_{i}^{2})\overline{P}\| \leq \overline{\lambda}^{2}  \overset{(a)}{\Longleftrightarrow}\lambda_{i}^{2} \sum_{j=1}^{M_R}\overline{P}_{ij} \leq \overline{\lambda}^{2},  \nonumber \\
&\overset{(b)}{\Longleftrightarrow} \sum_{j=1}^{M_R} \sum_{s, p}P_{ij}(s, p)\mathbb{P}(s, p~|~r_j) \leq \frac{\overline{\lambda}^{2}}{\lambda_{i}^{2}} \triangleq \theta_{i}^{2},  \label{eq: s-c}
\end{align}
where $P_{ij}(s, p)$ is the conditional probability defined in SD-MC model \eqref{sys:ss-MC}, and $\lambda_{i}, \overline{\lambda}$ are parameters defined in Assumption \ref{assumption}. 
By Bayes's Law, the conditional probability $\mathbb{P}(s, p~|~r_j)$ can be rewritten  in terms of decision variables $\{X(s, p, r)\}$.
\begin{align}
\mathbb{P}(s, p~|~r_j) &= \frac{\mathbb{P}(s, p)}{\mathbb{P}(r_j)} = \frac{\mathbb{P}(s, p)}{\sum_{s, p}\mathbb{P}(s, p, r_j)} \nonumber \\
&= \frac{X(s, p)}{\sum_{s, p}X(s, p, r_j)} = \frac{\sum_{r_j}X(s, p, r_j)}{\sum_{s, p}X(s, p, r_j)} \label{eq: p-x}
\end{align}
By replacing $\mathbb{P}(s, p~|~r_j)$ in \eqref{eq: s-c} with its representation in \eqref{eq: p-x}, the constraint \eqref{eq: s-c} can be rewritten $\forall 1 \leq i \leq M_R$ as  
\begin{align}
	\sum_{j=1}^{M_R} \sum_{s, p}P_{ij}(s, p)\frac{\sum_{r_j}X(s, p, r_j)}{\sum_{s, p}X(s, p, r_j)} \leq \theta_{i}^{2}, \quad  \label{ineq: X}
\end{align}
Let $X(r_j):=\sum_{s, p}X(s, p, r_j) > 0$ and $X(s, p):=\sum_{r_j}X(s, p, r_j)$, then multiplying $\prod_{j=1}^{M_R}X(r_j) > 0$ on both sides of inequality \eqref{ineq: X} leads to 
\begin{align*}
	&\sum_{j=1}^{M_R} \sum_{s, p}P_{ij}(s, p)\frac{X(s, p)}{X(r_j)}  \prod_{j=1}^{M_R}X(r_j) \leq \theta_{i}^{2} \prod_{j=1}^{M_R}X(r_j), \forall i \\
	&\Longleftrightarrow \sum_{j=1}^{M_R}\sum_{s, p}P_{ij}(s, p)X(s, p)\prod_{\ell \neq j}X(r_{\ell}) \leq \theta_{i}^{2} \prod_{j=1}^{M_R}X(r_j), \forall i.
\end{align*}
This thus show that the solutions of CPOP \eqref{opt: qcqp} is equivalent to those of the original optimization  in Problem \ref{problem-efficiency}. %The proof is comple.
%It follows that $X(s, r)={\rm Pr}\{s, r\}=\sum_{a \in A(s), p \in \Omega_{p}}{\rm Pr}\{s, r, a, p\}=\sum_{a \in A(s), p \in \Omega_{p}}X(s, r, a, p)$. Then, since $X(r, s, a)={\rm Pr}\{r, s, a\}={\rm Pr}\{a | r, s\}{\rm Pr}\{r, s\}={\rm Pr}\{a | r, s\}X(r, s)$, one has ${\rm Pr}\{a | r, s\}=X(r, s, a)\big/X(r, s)=\sum_{p \in \Omega_{p}}X(s, r, a, p)\big/\sum_{a \in A(s), p \Omega_{p}}X(s, r, a, p)$. Similarly, one has ${\rm Pr}\{p | r, s\}=\sum_{a \in A(s)}X(s, r, a, p)\big/ \sum_{a \in A(s), p \in \Omega_{p}}X(s, r, a, p)$. 
\end{proof}

%\newpage
\begin{proof}[Proof of Lemma \ref{lemma:qcp}]
Consider the polynomial constraint in \eqref{opt: stability} with a two-state SD-MC (i.e., $M_R = 2$).
Then the polynomial inequality \eqref{opt: stability} is reduced to
\begin{align}
	&\sum_{j=1}^{2}\sum_{s, p}P_{ij}(s, p)X(s, p)\prod_{\ell \neq j}X(r_{\ell}) \leq \theta_{i}^{2} \prod_{j=1}^{2}X(r_j), \forall i=1, 2. \nonumber \\
	&\Longleftrightarrow \sum_{s, p}P_{i1}(s, p)X(s, p)X(r_{2}) + \sum_{s, p}P_{i2}(s, p)X(s, p)X(r_{1}) \nonumber\\
	 \leq & \theta_{i}^{2}X(r_1)X(r_2), \forall i=1, 2. \label{ineq:q-term}
\end{align}
It is clear that the order of the polynomial terms in the inequality \eqref{ineq:q-term} is $2$, which makes it a quadratic constraint. 
Let $X=[X(r_1, s_1, p_1), \ldots, X(r_2, s_1, p_1), \ldots, X(r_2, s_{M_s}, p_{M_p})]^{T}$ denotes a column vector with $M_s=|S|$ and $M_p=|\Omega_{p}|$ being the sizes of the MDP state set and transmission power set, respectively.
Let $I_{1}=[\mathbf{e}, \mathbf{0}]$, $I_{2}=[\mathbf{0},\mathbf{e}]$ with $\mathbf{0}_{1\times M_{s}M_{p}}$ and $\mathbf{e}_{1\times M_{s}M_{p}}$ are the row vectors of $M_{s}M_{p}$ zeros and ones, respectively. 
Let $\vec{P}_{ij}=Vec(P_{ij})=[P_{ij}(s_{1}, p_{1}), P_{ij}(s_{2}, p_{1}), \ldots, P_{ij}(s_{M_s}, p_{1}), \ldots, P_{ij}(s_{M_s}, p_{M_p})]$ denotes the vectorization of the matrix $P_{ij}, \forall i,j=1,2$. Then, with $X(r_1)=I_{1}X$ and $X(r_2)=I_{2}X$, the quadratic constraint in \eqref{ineq:q-term} can be rewritten as 
\begin{align*}
	&X^{T}\begin{bmatrix}
		\vec{P}_{i1}^{T} \\
		\vec{P}_{i1}^{T}
	\end{bmatrix}I_{2}X + x^{T}\begin{bmatrix}
	\vec{P}_{i2}^{T} \\
	\vec{P}_{i2}^{T}
\end{bmatrix}I_{1}X - \theta_{i}^{2}I_{1}XI_{2}X \leq 0, \forall i=1, 2 \\
&\Longleftrightarrow X^{T}\bigg(\underbrace{\begin{bmatrix}
	\vec{P}_{i1}^{T} \\
	\vec{P}_{i1}^{T}
\end{bmatrix}I_{2} +\begin{bmatrix}
\vec{P}_{i2}^{T} \\
\vec{P}_{i2}^{T}
\end{bmatrix}I_{1}- \theta_{i}^{2}I_{1}^{T}I_{2}}_{\overline{Q}_{i}}\bigg)X \leq 0, \forall i=1, 2
\end{align*}
The proof is then complete. 
\end{proof}

\begin{proof}[Proof of Proposition \ref{prop:lpr}]
	Properties 1) and 3) hold due to the fact that for any joint-policies $\mu $ that satisfy the $1$-norm constraint in \eqref{ineq:norm-1}, they guarantee that the $\infty$-norm stability condition \eqref{ineq:condition-lambda} is satisfied. 
	This implies that the feasible set of the joint policies $\mathcal{H}_{\tilde{\mu}}$ induced by the $1$-norm stability constraint in \eqref{ineq:norm-1} is strictly smaller than the one $\mathcal{H}_{\mu}$ generated by the $\infty$-norm condition in \eqref{ineq:condition-lambda}. 
	Therefore, the optimal costs $J^{*}$ under the set $\mathcal{H}_{\mu}$ are clearly smaller than the optimal costs under the set $\mathcal{H}_{\tilde{\mu}}$. 
	To prove  property 2), let us consider the $1$-norm condition as follows.
	\begin{align}
		&\overline{\lambda} > \sqrt{M_R\|\diag(\lambda_{i}^{2})\overline{P}(\mu)\|_{1}} \nonumber\\
		\Longleftrightarrow & \|\diag(\lambda_{i}^{2})\overline{P}(\mu)\|_{1} < \bigg(\frac{\overline{\lambda}}{M_R}\bigg)^{2} \nonumber\\
		\Longleftrightarrow & \sum_{i=1}^{M_R}\lambda_{i}^{2}\sum_{s, p}P_{ij}(s, p)\mathbb{P}(s, p~|~r_j) < \bigg(\frac{\overline{\lambda}}{M_R}\bigg)^{2}, \forall j \label{ineq:norm-1-new}
	\end{align}
By \eqref{eq: p-x}, the inequality in \eqref{ineq:norm-1-new} may then be rewritten as 
\begin{align}
\sum_{i=1}^{M_R}\lambda_{i}^{2}\sum_{s, p}P_{ij}(s, p) \frac{\sum_{r_j}X(s, p, r_j)}{\sum_{s, p}X(s, p, r_j)} <  \bigg(\frac{\overline{\lambda}}{M_R}\bigg)^{2}, \forall j \label{ineq:1-norm-x}
\end{align}
Multiplying $\sum_{s, p}X(s, p, r_j)$ on both sides of the inequality in \eqref{ineq:1-norm-x} implies for all $j$ that
\begin{align*}
	\sum_{i=1}^{M_R}\lambda_{i}^{2}\sum_{s, p}P_{ij}(s, p)\sum_{r_j}X(s, p, r_j) < \sum_{s, p}X(s, p, r_j)\bigg(\frac{\overline{\lambda}}{M_R}\bigg)^{2}
\end{align*}
It is clear from the above inequality that the $1$-norm stability conditions in \eqref{ineq:norm-1} are thus linear constraints. 
Hence, the optimization problem formulated in \eqref{opt: game-1} is a LP under the relaxed $1$-norm stability conditions. The proof is complete. 
\end{proof}
\balance
\bibliographystyle{IEEEtran}       % Include this if you use bibtex 
\bibliography{bibfile/TAC2020} 

\begin{IEEEbiography}{Bin Hu}
Bin Hu received the M.S. degree in control and system engineering from Zhejiang University, Hangzhou, China, in 2010 and the Ph.D. degree from The University of Notre Dame, Notre Dame, IN, USA, in 2016. He is currently an Assistant Professor in the department of engineering technology at Old Dominion University, Norfolk, VA, USA. His research interests include networked control systems, cybersecurity, distributed control and optimization, and human machine interaction.
\end{IEEEbiography}

\begin{IEEEbiography}{Tua A. Tamba}
Tua A. Tamba received the B.Eng. in   Engineering Physics from Institut Teknologi Bandung, Indonesia in 2006, the M.Sc. in Mechanical Engineering from Pusan National University, Republic of Korea in 2009, and both the MSEE and Ph.D. in Electrical Engineering from University of Notre Dame, USA in 2016. He is currently an assisstant professor in the Department of Electrical Engineering at Parahyangan Catholic University, Indonesia. His main research interests include dynamical systems, control theory, and optimization.	
\end{IEEEbiography}
\end{document}